\documentclass[11pt,english]{article} 
\usepackage{hyperref}
\usepackage{verbatim}
\usepackage{amsfonts}
\usepackage{amsmath,amssymb,amsthm}
\usepackage{multirow}
\usepackage{url}

\usepackage{amsfonts}
\usepackage{courier}
\usepackage{fullpage}

\usepackage{epsfig}
\usepackage[usenames]{color}
\usepackage[]{algorithm2e}
\usepackage{xspace}
\usepackage{times}
\usepackage{enumerate} 

\usepackage[utf8]{inputenc}
\usepackage{ucs}
\usepackage{tikz}
\usepackage{tkz-tab}
\usepackage{caption}
\usepackage{subcaption}
\usepackage{latexsym}
\usepackage{pgfplots}
\pgfplotsset{compat=newest} 
\pgfplotsset{plot coordinates/math parser=false}

\newtheorem{theorem}{Theorem}[section]
\newtheorem{lemma}[theorem]{Lemma}

\newtheorem{fact}[theorem]{Fact}

\newtheorem{definition}[theorem]{Definition}

\renewcommand{\paragraph}[1]{\noindent {\bf #1}}
\newcommand{\Exp}{\operatornamewithlimits{\mathbb{E}}}

\newcommand{\eps}{\epsilon}

\newcommand{\F}{\mathbb{F}}

\newcommand{\polylog}[1]{\mathrm{polylog}(#1)}
\newcommand{\poly}[1]{\mathrm{poly}(#1)}
\newcommand{\dist}[1]{\mathrm{dist}(#1)}
\newcommand{\wt}[1]{\mathrm{wt}(#1)}
\newcommand{\ith}[1]{{#1}^{\text{th}}}

\newcommand{\RS}{{\mathrm{RS}}}
\newcommand{\GV}{{\mathrm{GV}}}

\newcommand{\dmin}{d_{\textrm{min}}}

\newcommand{\ignore}[1]{}

\newcommand {\etal} {\textit{et al.}\xspace}

\makeatletter
\renewcommand{\algocf@caption@boxruled}{%
  \hrule
  \hbox to \hsize{%
    \vrule\hskip-0.4pt
    \vbox{   
       \vskip\interspacetitleboxruled%
       \unhbox\algocf@capbox\hfill
       \vskip\interspacetitleboxruled
       }%
     \hskip-0.4pt\vrule%
   }\nointerlineskip%
}%
\makeatother

\begin{document}
\title{A New Coding-based Algorithm for Finding Closest Pair of Vectors\thanks{
An extended abstract of this article appeared in 
Proceedings of the 13th International Computer Science Symposium in Russia (CSR'18), pages 321--333.}}
\author{Ning Xie\thanks{Florida International University, Miami, FL 33199, USA. Email: {\texttt nxie@cis.fiu.edu}. Research supported in part by NSF grant 1423034.} 
\and Shuai Xu\thanks{Florida International University, Miami, FL 33199, USA. Email: {\texttt sxu010@fiu.edu}.
Research supported in part by NSF grant 1423034.} 
\and Yekun Xu\thanks{Florida International University, Miami, FL 33199, USA. Email: {\texttt yxu040@fiu.edu}.
Research supported in part by NSF grant 1423034.}}
\date{}
\setcounter{page}{0}
\maketitle

\begin{abstract}
Given $n$ vectors $x_0, x_1, \ldots, x_{n-1}$ in $\{0,1\}^{m}$, 
how to find two vectors whose pairwise Hamming distance is minimum?
This problem is known as the \emph{Closest Pair Problem}. 
If these vectors are generated uniformly at random  except two of them are
correlated with Pearson-correlation coefficient $\rho$, then the problem is called
the \emph{Light Bulb Problem}.
In this work, we propose a novel coding-based scheme for the Closest Pair Problem. 
We design both randomized and deterministic algorithms, 
which achieve the best-known running time when the length of input vectors $m$ is small 
and the minimum distance is very small compared to $m$. 
Specifically, the running time of our randomized algorithm is $O(n\log^{2}n\cdot 2^{c m} \cdot \poly{m})$ 
and the running time of our deterministic algorithm is $O(n\log{n}\cdot 2^{c' m} \cdot \poly{m})$, 
where $c$ and $c'$ are constants depending only on the (relative) distance of the closest pair.
When applied to the Light Bulb Problem, our result yields state-of-the-art deterministic running time 
when the Pearson-correlation coefficient $\rho$ is very large. 
Specifically, when  $\rho \geq 0.9933$, our deterministic algorithm runs faster 
than the previously best deterministic algorithm (Alman, SOSA 2019).

\end{abstract}

{\bf keywords:} Closest Pair Problem, Light Bulb Problem, Error Correcting Codes

\newpage

\section{Introduction}
We consider the following classic \emph{Closest Pair Problem}: 
given $n$ vectors $x_0, x_1, \ldots,$ $x_{n-1}$ in $\{0,1\}^{m}$, 
how to find the two vectors with the minimum pairwise distance? 
Here the distance is the usual Hamming distance: 
$\dist{x_i, x_j}=|\{k \in [m]: (x_i)_k \neq (x_j)_k\}|$, 
where $(x_i)_k$ denotes the $\ith{k}$ component of vector $x_i$. 
Without loss of generality, we assume that $\dmin=\dist{x_0, x_1}$ is the unique minimum distance 
and all other pairwise distances are greater than $\dmin$.

The Closest Pair Problem is one of the most fundamental and well-studied problems in many science disciplines, 
having a wide spectrum of applications in computational finance, DNA detection, weather prediction, etc. 
For instance, the Closest Pair Problem has the following interesting application in bioinformatics. 
Scientists wish to find connections between Single Nucleotide Polymorphisms (SNPs) and phenotypic traits. 
SNPs are one of the most common types of genetic differences among people,
with each SNP representing a variation in a single DNA block called \emph{nucleotide}~\cite{FBC+07}. 
Screening for most correlated pairs of SNPs has been applied to
study such connections~\cite{ARL+05, CNG+98, Cor09, MSL+07}. 
As the number of SNPs in humans is estimated to be around 10 to 11 million, 
for problem size $n$ of this size, any improvement in running time
for solving the Closest Pair Problem would have huge impacts on genetics and computational biology~\cite{MSL+07}. 

In theoretical computer science, the Closest Pair Problem 
has a long history in computational geometry, 
see e.g.~\cite{Smi97} for a survey of many classic algorithms for the problem. 
The naive algorithm for the Closest Pair Problem takes $O(mn^2)$ time. 
When the dimension $m$ is a constant, either in the Euclidean space or $\ell_p$ space, 
the classic divide-and-conquer based algorithm runs in $O(n \log n)$ time~\cite{Ben80}. 
Rabin~\cite{Rab76} combined the floor function with randomization to devise a linear time algorithm.   
In 1995, Khuller and Matias~\cite{KM95} simplified Rabin's algorithm 
to achieve the same running time $O(n)$ and space complexity $O(n)$. 
Golin~\etal~\cite{GRSS95} used dynamic perfect hashing to implement a dictionary 
and obtained the same linear time and space bounds.

When the dimension $m$ is not a constant, due to a well-know phenomenon called \textit{curse of dimensionality}, 
this problem becomes much harder. 
The first subquadratic time algorithm for the Closest Pair Problem is due to Alman and Williams~\cite{AW15} for $m$ as large as $\log^{2-o(1)}{n}$.
The algorithm is built on a newly developed framework called 
\emph{polynomial method}~\cite{Wil14a,Wil14b,AWY15}. 
In particular, Alman and Williams first constructed a probabilistic polynomial of degree 
$O(\sqrt{n\log{1/\epsilon}})$ which computes the \textsc{MAJORITY} function on $n$ variables with error at most $\epsilon$, 
then applied the polynomial method to design an algorithm which runs in $n^{2-1/O(s(n)\log^2 s(n))}$ time
where $m=s(n)\log{n}$, and computed the minimum Hamming distance among all red-blue vector pairs\footnote{The actual problem solved in~\cite{AW15} is the so-called \emph{Bichromatic Hamming Closest Pair Problem}; see discussion in Section~\ref{sec:related_work} below.} through polynomial evaluations.
In a more recent work, Alman~\etal~\cite{ACW16} unified Valiant's fast matrix multiplication 
approach~\cite{Val15} with that by Alman and Williams~\cite{AW15}.
They constructed probabilistic \emph{polynomial threshold functions} (PTFs) to obtain 
a simpler algorithm which improved to randomized time 
$n^{2-1/O(\sqrt{s(n)}\log^{3/2}{s(n)})}$ or deterministic time $n^{2-1/O(s(n)\log^2{s(n)})}$.

\medskip
\paragraph{The Light Bulb Problem.}
A special case of the Closest Pair Problem, the so-called \emph{Light Bulb Problem}, 
was first posed by L. Valiant in 1988~\cite{Val88}. 
In this problem, we are given a set of $n$ vectors in $\{0,1\}^m$ chosen uniformly at random 
from the Boolean hypercube, except that two of them are non-trivially correlated 
(specifically, have Pearson-correlation coefficient $\rho$, which is equivalent to that 
the expected Hamming distance between the correlated pair is $\frac{1-\rho}{2}m$), 
the problem then is to find the correlated pair. 

Paturi~\etal~\cite{PRR95} gave the first non-trivial algorithm, which runs\footnote{We adopt the common notation
$\tilde{O}(n^k)$ to denote $n^k\cdot \polylog{n}$.} in $\tilde{O}(n^{2-\log (1+\rho)})$. 
In 2010, Dubiner~\cite{Dub10} proposed a Bucketing Coding algorithm which runs in time 
$\tilde{O}(n^{\frac{2}{1+\rho}})$. 
The well-known \emph{locality sensitive hashing} scheme of Indyk and Motwani~\cite{IM98} 
performs slightly worse than Paturi~\etal's hash-based algorithm but recent \emph{data-dependent LSH}~\cite{ALRW17}
matches the running time of Dubiner's. Roughly speaking, a family of hash functions $\mathcal{H}$ is 
called \emph{$(r, cr, p_1, p_2)$-sensitive} if, for any two points $p$ and $q$ in a metric space $(X, d)$, a randomly chosen hash function
from $\mathcal{H}$ hashes $p$ and $q$ into the same bucket with probability at least $p_1$ if they are close (i.e., when $d(p,q)\leq r$) and
with probability at most $p_2$ if they are far apart (i.e., when $d(p,q)\geq cr$), where $c>1$ is the \emph{approximation} factor and $p_1>p_2$.
Indyk and Motwani~\cite{IM98} proved that such a family of LSH can be used to construct a data structure solving the $c$-approximate
Nearest Neighbor Search problem. Specifically,  for a data set consisting of at most $n$ points from $X$,
the data structure uses $\tilde{O}(n^{1+\varrho})$ space (and $\tilde{O}(n^{1+\varrho})$ preprocessing time)
and supports $\tilde{O}(m\cdot n^\varrho)$ query time, where $m$ is the dimension of the space and 
$\varrho:=\frac{\log{1/p_1}}{\log{1/p_2}}$ basically quantifies the quality of the LSH.
When $(X,d)$ is the Hamming space, the original work of Indyk and Motwani~\cite{IM98} achieved $\varrho \leq 1/c$,
while the current best result is $\varrho=1/2c-1$ by Andoni~\etal~\cite{ALRW17}, under the framework of \emph{data-dependent LSH}.
Applying LSH to the Light Bulb Problem, we have $m=O(\log{n})$, $c\geq\frac{1}{1-\rho}$ with high probability, and 
we need to pay the one-time preprocessing time and $n$ queries for each vector to search for its nearest neighbor in the data set.
Therefore LSH solves the Light Bulb Problem in time $\tilde{O}(n^{2-\rho})$ using the original data-independent scheme of Indyk and Motwani,
and can be improved to $\tilde{O}(n^{\frac{2}{1+\rho}})$ using the data-dependent scheme in~\cite{ALRW17}.
As $\rho$ gets small, all these three algorithms have running time 
$\tilde{O}(n^{2-c\rho})$ for various constants $c$.\footnote{When $\rho$ goes to zero, the exponent in the running time of
Paturi~\etal~\cite{PRR95} is $2-\log(e)\cdot\rho+O(\rho^2)$.} 
Comparing the constants in these three algorithms, 
Dubiner and data-dependent LSH achieve the best constant, which is $\tilde{O}(n^{2-2\rho})$, in the limit of $\rho \to 0$. 
Asymptotically the same bound was also achieved by May and Ozerov~\cite{MO15}, in which
the authors used algorithms that find Hamming closest pairs to improve the running time of
decoding random binary linear codes.
 
The breakthrough result of Valiant~\cite{Val15} is a fast matrix multiplication based algorithm 
which finds the ``planted'' closest pair in time 
$O(\frac{n^{\frac{5-\omega}{4-\omega}+\epsilon}}{\rho^{2\omega}})< n^{1.62} \cdot \poly{1/\rho}$ 
with high probability for any constant $\epsilon,\rho>0$ and $m > n^{\frac{1}{4-\omega}}/\rho^2$, 
where $\omega<2.373$ is the exponent of fast matrix multiplications. 
The most striking feature of Valiant's algorithm is that $\rho$ does not appear in the exponent of $n$ 
in the running time of the algorithm.
Karppa~\etal~\cite{KKK16} further improved Valiant's algorithm to $n^{1.582}$. Very recently, Alman~\cite{Alm19} combined techniques in~\cite{Val15} with the polynomial method to give a very elegant and simple algorithm which matches Karppa~\etal' bound. Moreover, Alman derandomized his algorithm and improved on the previously best
deterministic running time by Karppa~\etal~\cite{KKKO16}. 
Note that Valiant, Karppa~\etal and Alman achieved runtimes of $n^{2-\Omega(1)}(m/\epsilon)^{O(1)}$ 
for the Light Bulb Problem, which improved upon previous algorithms that rely on the
Locality Sensitive Hashing (LSH) schemes. The LSH based algorithms only achieved runtime of 
$n^{2-O(\epsilon)}$ for the Light Bulb Problem. 

We remark that all the above-mentioned algorithms (except May and Ozerov's work) 
that achieve state-of-the-art running time
are based on either involved probabilistic polynomial constructions or impractical $O(n^\omega)$ fast
matrix multiplications\footnote{Subcubic fast matrix multiplication algorithms are practical for Strassen-based ones~\cite{BB15,HRM+17} and are practical for very large input sizes up to $\omega=2.7734$ (see e.g. the survey~\cite{Pan18}). However, all other theoretically more efficient algorithms, such as recent developments~\cite{Sto10,Wil12,Leg12}, are superior to the trivial cubic algorithm only for matrices of colossal sizes.}, or both. 

\medskip
\paragraph{Overview of our main results.}
In this work, we propose a new coding-based scheme for the Closest Pair Problem. 
We design both randomized and deterministic algorithms, which achieve the best-known running time when the length of input vectors $m$ is small ($m=O(\log{n})$) and the minimum distance is very small compared to $m$. 
Specifically, the running time of our randomized algorithm is $O(n\log^{2}n\cdot 2^{c m} \cdot \poly{m})$ 
and the running time of our deterministic algorithm is $O(n\log{n}\cdot 2^{c' m} \cdot \poly{m})$, 
where $c$ and $c'$ are constants depending only on the (relative) distance of the closest pair; see Section~\ref{sec:our_results} for precise statements. Since the running time of our algorithms are exponential in $m$, they are subquadratic-time algorithms only when $m\leq \alpha \log{n}$ for some constant $\alpha>0$.
When applied to the Light Bulb Problem, our deterministic algorithm achieves state-of-the-art running time 
when the Pearson-correlation coefficient $\rho$ is very large.

\subsection{Our approach}
\RestyleAlgo{boxruled}
\LinesNumbered
\begin{algorithm}[ht] 
\caption{General Idea of Main Algorithm}\label{alg:main_intro}
\SetKwData{Left}{left}\SetKwData{This}{this}\SetKwData{Up}{up}
\SetKwFunction{Decode}{Decode}
\SetKwInOut{Input}{input}\SetKwInOut{Output}{output}
\Input{A set of $n$ vectors $x_0,  \ldots, x_{n-1}$ in $\{0,1\}^{m}$ and $\dmin$ }
\Output{Two vectors and their distance}
\BlankLine
	generate a binary code $C \subseteq \{0,1\}^{m}$  \\
	pick a random $y\in \{0,1\}^{m}$ \\
	\For{$j\leftarrow 0$ \KwTo $n-1$}{\label{forins'}
		decode $y+x_j$ in $C$, and denote the resulting vector by $\tilde{x}_j$ 
	}
	sort $\tilde{x}_{0}, \ldots, \tilde{x}_{n-1}$ \\ 
		
	\For {\emph{each of the $n-1$ pairs of adjacent vectors in the sorted list}} {
	    compute the distance between the two \emph{original} vectors.
	}
	output the pair of vectors with the minimum distance and their distance \
\end{algorithm}

We propose a simple, error-correcting code based scheme for the Closest Pair Problem.
Apart from achieving the best running time for certain range of parameters,
we believe that our new approach has the merit of being simple, and hence more likely being practical as well.
In particular, neither complicated data structure nor fast matrix multiplication is employed in our algorithms.

The basic idea of our algorithms is very simple. 
Suppose for concreteness that $x_0$ and $x_1$ are the unique pair of vectors that achieve the minimum distance. 
Our scheme is inspired by the extreme case when $x_0$ and $x_1$ are identical vectors.
In this case, a simple \emph{sort and check} approach solves the problem in $O(mn\log{n})$ time:
sort all $n$ vectors and then compute only the $n-1$ pairwise distances (instead of all
$\binom{n}{2}$ distances) of adjacent vectors in the sorted list.
Since the two closest vectors are identical, they must be adjacent in the sorted list and thus the algorithm
would compute their distance and find them. 
This motivates us to view the input vectors as received messages that were encoded by an error correction code
and have been transmitted through a noisy channel. 
As a result, the originally identical vectors are no longer the same, nevertheless are still very close.
Directly applying the \emph{sort and check} approach would fail but a natural remedy is 
to decode these received messages into codewords first. 
Indeed, if the distance between $x_0$ and $x_1$ is small and we are lucky to 
have a codeword $c$ that is very close to both of them, 
then a unique decoding algorithm would decode both of these two vectors into $c$. 
Now if we ``sort'' the \emph{decoded} vectors and 
then ``check'' the corresponding \emph{original} vectors of each adjacent pair of vectors\footnote{Actually, we only need to ``check'' when the two adjacent decoded vectors are identical.}, 
the algorithm would successfully find the closest pair.
How to turn this ``good luck'' into a working algorithm? 
Simply try different shift vectors $y$ and view $y+x_i$ as the input vectors,
since the Hamming distances are invariant under any shift.
The basic idea of our approach is summarized in Algorithm~\ref{alg:main_intro}.

Figure~\ref{fig1} illustrates the effects ``bad'' shift vectors and ``good'' shift vectors
on the decoding part of our algorithm; here single arrows with dotted lines point from original vectors to shifted target vectors, while
double arrows indicate into which codewords are shifted target vectors decoded. 
In Figure~\ref{fig:1A}, our shifted target vectors $x_0+y'$ and $x_1+y'$ are decoded into two different codewords, so $y'$ is a bad shift. 
In Figure~\ref{fig:1B}, our shifted target vectors $x_0+y$ and $x_1+y$ are decoded into the same codeword, 
therefore we can apply the sort-and-check approach to find the closest pair.

Figure~\ref{fig2} illustrates what happens if we sort the vectors directly and why sorting decoded vectors
works.


\begin{figure}
     \begin{subfigure}{0.45\textwidth}
          \centering
          \includegraphics[scale=0.8]{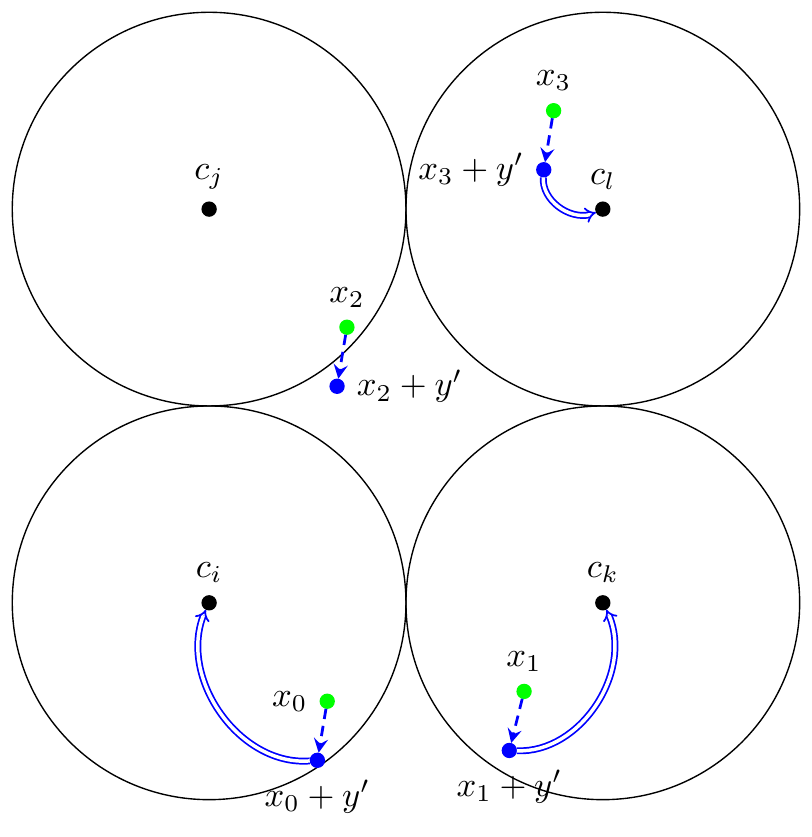} 
          \caption{bad shift}
          \label{fig:1A}
     \end{subfigure}
     \qquad
     \begin{subfigure}{0.45\textwidth}
          \centering
          \includegraphics[scale=0.8]{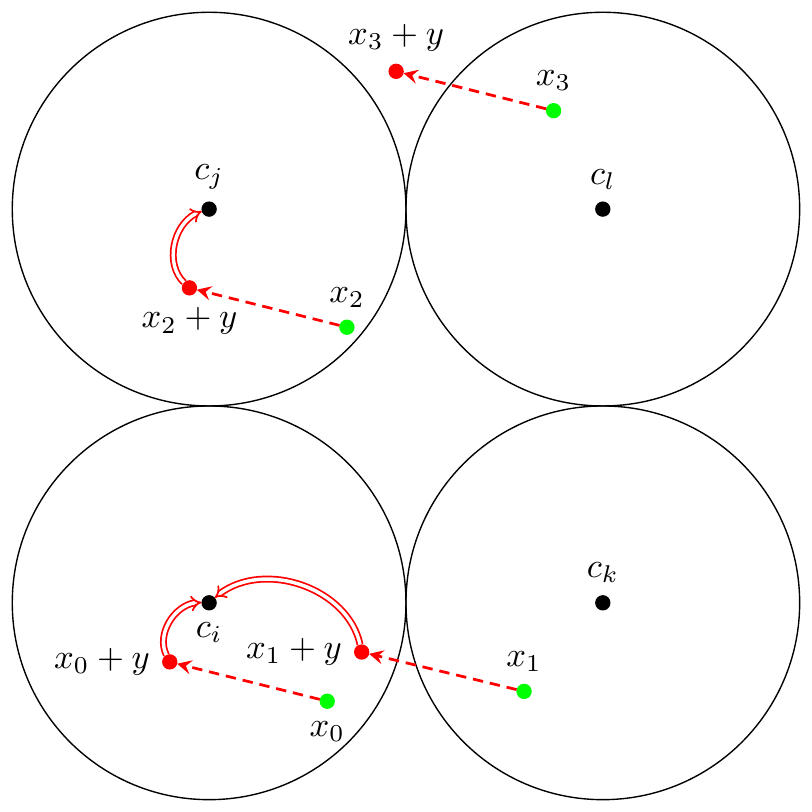} 
          \caption{good shift}
          \label{fig:1B}
     \end{subfigure}
     \caption{Decoding with good and bad shift vectors}\label{fig:M1}\label{fig1}
\end{figure}

\begin{figure}
     \begin{subfigure}{0.45\textwidth}
          \centering
          \includegraphics[scale=0.75]{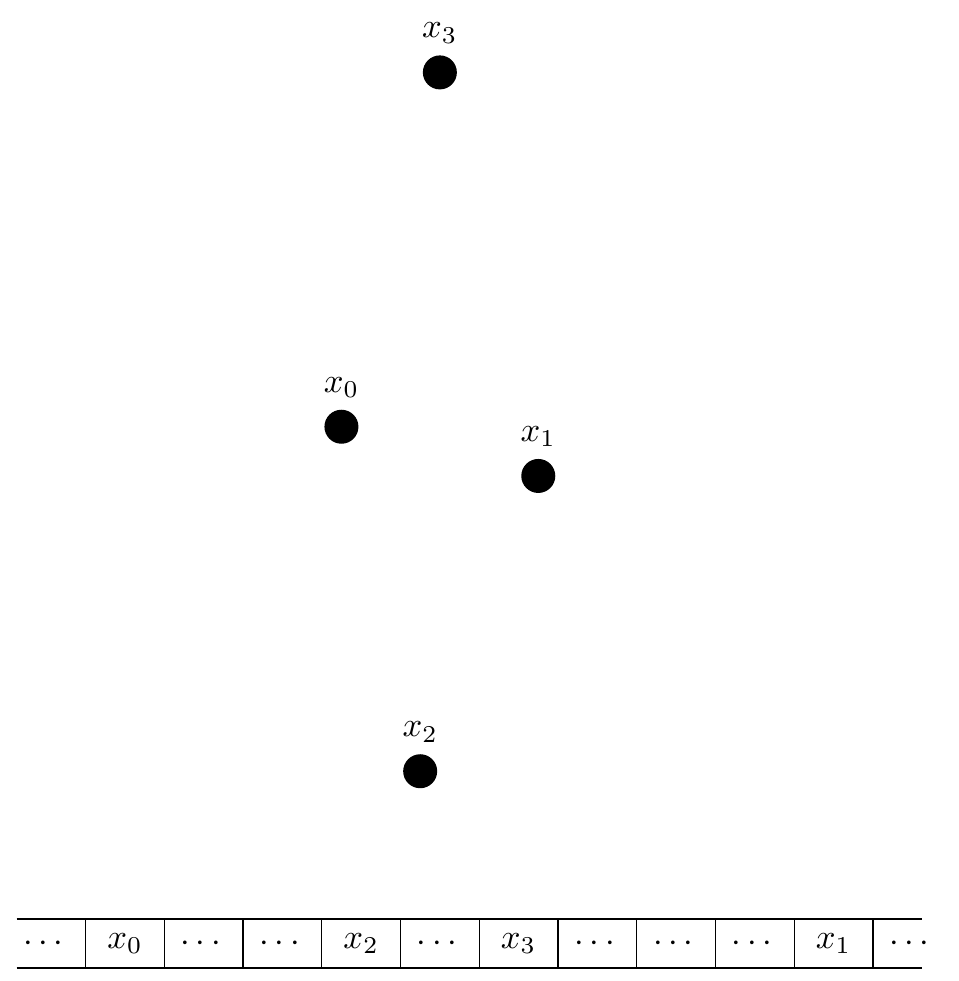} 
          \caption{Sorting original vectors directly}
          \label{fig:2A}
     \end{subfigure}
     \qquad
     \begin{subfigure}{0.45\textwidth}
          \centering
          \includegraphics[scale=0.75]{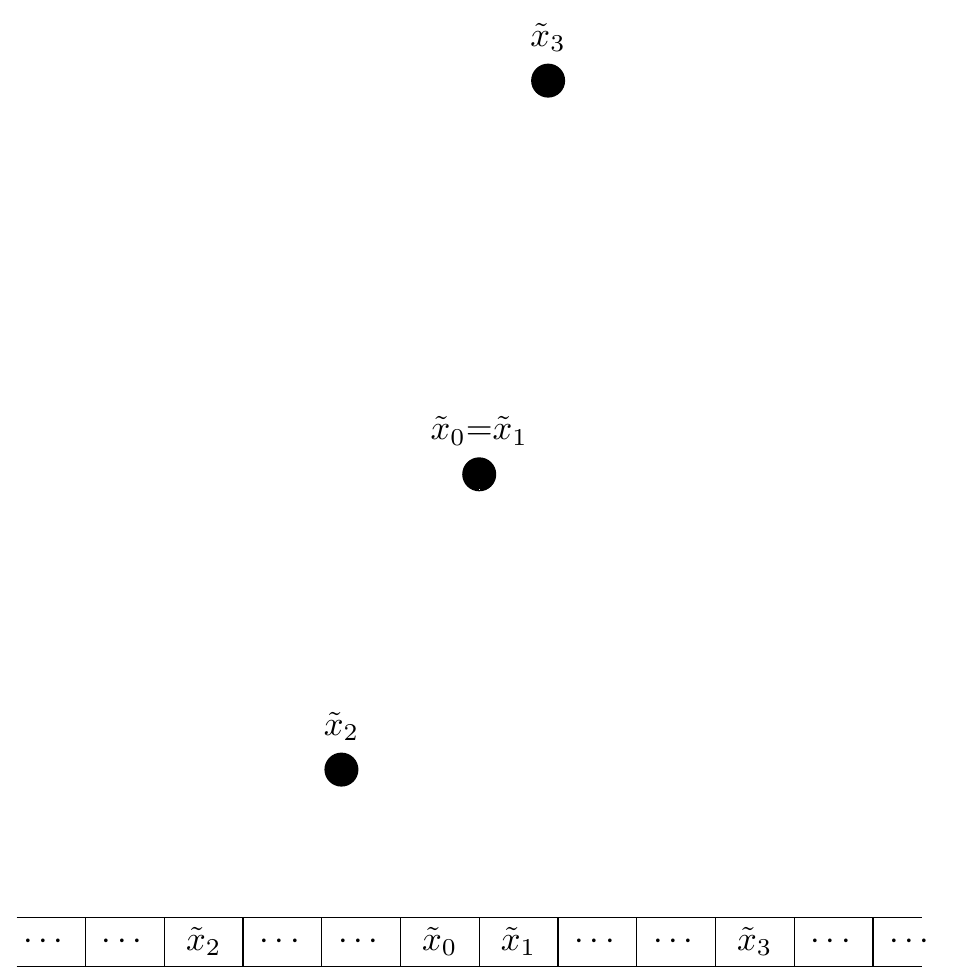}  
          \caption{Sorting decoded vectors}
          \label{fig:2B}
     \end{subfigure}
     \caption{Difference between sorting input vectors directly and sorting decoded vectors.} \label{fig2}
\end{figure}

Making the idea of decoding work for larger minimum pairwise distance involves balancing the parameters of 
the error-correcting code so that it is efficiently decodable 
as well as having appropriate decoding radius. 
The decoding radius $r$ should have the following properties. 
On one hand, $r$ should be small to ensure that there is a codeword $c$ such that 
only $x_0$ and $x_1$ will be decoded into $c$ 
(therefore $x_0$ and $x_1$ will be adjacent in the sorted array and hence will be compared with each other). 
On the other hand, we would like $r$ to be large so as to maximize the number of ``good'' shift vectors 
which enable both $x_0$ and $x_1$ decoding to the same codeword. 
As a result, our algorithms generally perform best when the closest pair distance is very small.


\subsection{Our results}\label{sec:our_results}
Our simple error-correcting code based algorithm can be applied to  
solve the Closest Pair Problem and the Light Bulb Problem.

\medskip
\subsubsection{The Closest Pair Problem} 
Our main result is the following simple randomized algorithm for the Closest Pair Problem.
\begin{theorem}[Main]\label{thm:general_rand'}
Let $x_0, x_1, \ldots, x_{n-1}$ in $\{0,1\}^{m}$ be $n$ binary vectors such that 
$x_0$ and $x_1$ is the unique pair achieving the minimum pairwise distance $\dmin$ (and the second smallest
distance can be as small as $\dmin + 1$). 
Suppose\footnote{In fact this assumption can be easily removed with a small overhead in the running time; 
see the discussion below and Section~\ref{sec:search}.} we are given the value of $\dmin$ and let $\delta := \dmin/m$.
Then there is a randomized algorithm
running in $O(n\log^{2}n\cdot 2^{(1-\kappa_{\mathsmaller{Z}}(\delta)-\delta)m} \cdot \poly{m})$ 
which finds the closest pair $x_0$ and $x_1$ with probability at least $1-1/n^2$.
The running time can be improved to 
$O(n\log^{2}n\cdot 2^{(H_{2}(\delta)-\delta)m} \cdot \poly{m})$,
if we are given black-box decoding algorithms for an ensemble of $O(\log m/\eps)$
binary error-correcting codes that meet the Gilbert-Varshamov bound.  
\end{theorem}
Here $\kappa_{\mathsmaller{GV}}(\delta)$ and $\kappa_{\mathsmaller{Z}}(\delta)$ are functions 
derived from the Gilbert-Varshamov (GV) bound and the Zyablov bound respectively 
(see Section~\ref{sec:our_codes} for details). 
Specifically, $\kappa_{\mathsmaller{GV}}(\delta)=1-H_{2}(\delta)$,
and both $\kappa_{\mathsmaller{GV}}(\delta)$ and $\kappa_{\mathsmaller{Z}}(\delta)$ are monotone decreasing functions
for $\delta \in [0, 1/2]$, with function values ranging from $1$ to $0$; see e.g. Figure~9.2 in~\cite{GRS18} for an illustration.

The running time of our algorithm depends on --- in addition to the number of vectors $n$ --- 
both dimension $m$ and $\delta := \dmin/m$. 
To illustrate its performance we choose two typical 
vector lengths $m$, namely those corresponding to the Hamming bound\footnote{The Hamming bound, also 
known as the sphere packing bound, specifies an upper bound on the number of codewords 
a code can have given the block length and the minimum distance of the code.}
and the Gilbert-Varshamov (GV) bound\footnote{The GV bound is known to be attainable by random codes.},
and list the exponents $\gamma'$ in the running time of the GV-code version of our algorithm
as a function of $\dmin$ (in fact $\delta$) 
in Table~\ref{table-1}.
Here, we write the running of the algorithm as $\tilde{O}(n^{\gamma'})$,
where $\tilde{O}$ suppresses any polylogarithmic factor of $n$. 
One can see that our algorithm runs in subquadratic time when $\delta$ is small,
or equivalently when the Hamming distance between the closest pair is small. For instance, when $\delta=0.05$, and the length $m=1.4013 \log{n}$, then the running time is $O(n^{1.3313})$ if we use GV bound.

\begin{table}
\centering
\caption{Running time of our algorithm when vector length $m$ and relative distance $\delta$ meets the Hamming bound and GV bound }
\label{table-1}
\begin{tabular}{|c|c|c|c|c|}
\hline
         & \multicolumn{2}{c|}{Hamming bound}                                                                                                             & \multicolumn{2}{c|}{GV bound}                                                                                                                  \\ \hline
$\delta$ & \begin{tabular}[c]{@{}c@{}}length of vector\\ ($m/\log{n}$)\end{tabular} & \begin{tabular}[c]{@{}c@{}}exponent ($\gamma'$) \end{tabular} & \begin{tabular}[c]{@{}c@{}}length of vector\\ ($m/\log{n}$)\end{tabular} & \begin{tabular}[c]{@{}c@{}}exponent ($\gamma'$) \end{tabular} \\ \hline
0.01     & 1.0476                                                                   & 1.0742                                                               & 1.0879                                                                    & 1.0770                                                              \\ \hline
0.025     & 1.1074                                                                   & 1.1591                                                             & 1.2029                                                                    & 1.1728                                                              \\ \hline
0.05     & 1.2029                                                                   & 1.2844                                                               & 1.4013                                                                    & 1.3313                                                              \\ \hline
0.075     & 1.2999                                                                   & 1.4021                                                              & 1.6242                                                                   & 1.5024                                                              \\ \hline
0.1		& 1.4013                                                                   & 1.5171                                                              & 1.8832                                                                   & 1.6949                                                              \\ \hline
0.125     & 1.5090                                                                   & 1.6316                                                              & 2.1909                                                                   & 1.9170                                                              \\ \hline
0.133     & 1.5449 											 & 1.6684                                                               & 2.3064											 & 1.9989         
 \\ \hline
\end{tabular}
\end{table}

In the setting of $m=c\log{n}$ for some not too large constant $c$, 
the current best result is the randomized algorithm of Alman~\etal~\cite{ACW16}, which runs in $n^{2-1/O(\sqrt{c}\log^{3/2}{c})}$ time
for the Closest Pair Problem. As it is very hard to calculate the hidden constant in the exponent 
of their running time, it is impossible to compare our running time with theirs quantitatively.

\medskip
\paragraph{Deterministic algorithm.}
By checking all shift vectors up to certain Hamming weight, 
our randomized algorithm can be easily derandomized to yield the following theorem.
\begin{theorem}\label{thm:general_det'}
Let $x_0, x_1, \ldots, x_{n-1}$ in $\{0,1\}^{m}$ be $n$ binary vectors such that 
$x_0$ and $x_1$ is the unique pair achieving the minimum pairwise distance $\dmin$ (and the second smallest
distance can be as small as $\dmin + 1$). 
Suppose we are given the value of $\dmin$ and let $\delta := \dmin/m$.
Then there is a deterministic algorithm that finds the closest pair $x_0$ and $x_1$ with
running time $O(n\log{n}\cdot 2^{H_{2}(1-\kappa_{\mathsmaller{Z}}(\delta))m} \cdot \poly{m})$,
where $H_{2}(\cdot)$ is the binary entropy function.
Moreover, if we are given as black box the decoding algorithm of a random Varshamov linear code
with block length $m$ and minimum distance $\dmin+1$, then the running time is
$O(n\log{n}\cdot 2^{H_{2}(H_{2}(\delta))m} \cdot \poly{m})$.
\end{theorem}

\medskip
\paragraph{Searching for $\dmin$.}
If we remove the assumption that $\dmin$ is given, 
our algorithm can be modified to search for $\dmin$ first without too much slowdown; 
more details appear in Section~\ref{sec:search}.
\begin{theorem}
Let $x_0, x_1, \ldots, x_{n-1}$ in $\{0,1\}^{m}$ be $n$ binary vectors such that 
$x_0$ and $x_1$ is the unique pair achieving the minimum pairwise distance $\dmin$.
Then for any $\eps>0$, there is a randomized algorithm that
runs in $O(\eps^{-1}n\log^{2}n\cdot 
2^{(1-\kappa_{\mathsmaller{Z}}((1+\eps)\delta)-\delta H_{2}(\frac{1-\eps}{2}))m} \cdot \poly{m})$ 
which finds the $\dmin$ (and the pair $x_0$ and $x_1)$ with probability at least
$1-1/n$,
The running time can be improved to 
$O(\eps^{-1}n\log^{2}n\cdot 
2^{(H_{2}((1+\epsilon)\delta)-\delta H_{2}(\frac{1-\eps}{2}))m} \cdot \poly{m})$,
if we are given black-box decoding algorithms for an ensemble of $O(\log m/\eps)$
binary error-correcting codes that meet the Gilbert-Varshamov bound.
\end{theorem}

\medskip
\paragraph{Gapped version.}
Intuitively, if there is a gap between $\dmin$ and the second minimum distance, the Closest Pair Problem 
should be easier. This is reminiscent of the case of the $(1+\epsilon)$-Approximate NNS Problem
versus the NNS Problem.
However, as we still need to find the \emph{exact} solution to the Closest Pair Problem,
the situation here is different.

\begin{theorem}[Gapped version]\label{thm:general_rand2'}
Let $x_0, x_1, \ldots, x_{n-1}$ in $\{0,1\}^{m}$ be $n$ binary vectors such that 
$x_0$ and $x_1$ is the unique pair achieving the minimum pairwise distance $\dmin$.
Suppose we are given the values of $\dmin$ as well as the second minimum distance $d_2$.
Let $\delta := \dmin/m$ and $\delta' := d_2/m$.
Then there is a randomized algorithm
running in $O(n\log^{2}n\cdot 2^{(1-\kappa_{\mathsmaller{Z}}(\delta')
-\delta-(1-\delta)H_{2}(\frac{\delta' - \delta}{2(1-\delta)}))m} \cdot \poly{m})$ 
which finds the closest pair $x_0$ and $x_1$ with probability at least
$1-1/n^2$.  
Moreover, the running time can be further improved to 
$O(n\log^{2}n\cdot 2^{(H_{2}(\delta')
-\delta-(1-\delta)H_{2}(\frac{\delta'-\delta}{2(1-\delta)}))m} \cdot \poly{m})$, 
if we are given the black box access to the decoding algorithm of an $(m, K, d)$-code which meets
the Gilbert-Varshamov bound.
\end{theorem}

Our gapped version algorithm uses $d_2/2$ instead of $\dmin/2$ as the decoding radius. 
This, however, does not always give improved running time as illustrated in Figure~\ref{fig:3}.
In Figure~\ref{fig:3}, we set $\delta' = (1+\eps)\delta$ and write the running time as
$O(n\log^{2}n\cdot 2^{\gamma m} \cdot \poly{m})$ for both the gapped version (the blue line)
and the non-gapped version (the green line). 
One can see that using $d_2/2$ as the decoding radius does not always yield the best running time.
Indeed, this is the case only when $\eps$ is small enough. Our numerical calculations show that there exists an optimal decoding radius
$d_{opt}/2$ (which corresponds to the minimum point in the blue line) slightly larger than $\dmin/2$ such that
whenever $d_2\geq d_{opt}$ using $d_{opt}/2$ as the decoding radius will achieve the fastest running time.
Unfortunately we do not know how to calculate this $d_{opt}/2$ analytically.


\begin{figure}
	\centering
	\includegraphics[scale=1]{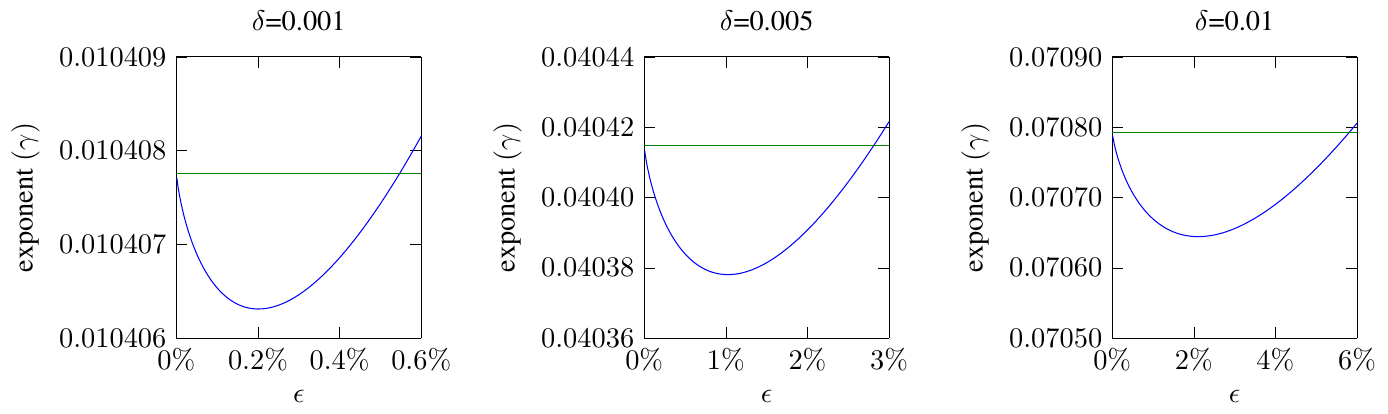}
	\caption{The range of $\epsilon$ in which gapped version outperforms non-gapped version}
	\label{fig:3}
\end{figure}

\subsubsection{The Light Bulb Problem}
Applying our algorithms for the Closest Pair Problem to the Light Bulb Problem easily yields the following theorem.
\begin{theorem}
There is a randomized algorithm for the Light Bulb Problem which runs in time
\[
O(n\cdot \poly{\log n})\cdot 2^{(1-\kappa_{\mathsmaller{Z}}(\frac{1-\rho}{2})-\frac{1-\rho}{2})
\frac{4\ln2\cdot \log n}{\rho^2}(1+o(1))}
\]
and succeeds with probability at least
$1-1/n^2$. 
The running time can be further improved to 
\[
O(n\cdot \poly{\log n})\cdot 2^{(H_{2}(\frac{1-\rho}{2})-\frac{1-\rho}{2})
\frac{4\ln2\cdot \log n}{\rho^2}(1+o(1))},
\]
if we are allowed a one-time preprocessing time\footnote{This is because 
the block length of the code is $m=4\ln{2}\log{n}/\rho^2 < 2.773 \log{n}/\rho^2$ and preprocessing the code 
requires $O(2^m)=O(n^{2.773/\rho^2})$ time.} of $n^{2.773/\rho^2}$ to generate the decoding
lookup table of a random Gilbert's $(m,K,(1-\rho)m/2)$-code.
Similar results can also be obtained for deterministic algorithms.
\end{theorem}

Our deterministic algorithm for the Light Bulb Problem performs faster than Alman's deterministic algorithm~\cite{Alm19}  when the Pearson-correlation coefficient $\rho$ is very large.
Moreover, we believe that our algorithms are very simple and therefore are likely to outperform other complicated ones
for at least not too large input sizes.

\subsection{Related work}\label{sec:related_work}

\paragraph{The Nearest Neighbor Search problem.}
The Closest Pair Problem is a special case of the more general \emph{Nearest Neighbor Search} (NNS) problem, 
defined as follows. Given a set $S$ of $n$ vectors in $\{0,1\}^m$, 
and a query point $q\in \{0,1\}^m$ as input, the problem is to find a point in $S$ which is closest to $q$. 
The performance of an NNS algorithm is usually measured by two parameters: 
the space (which is usually proportional to the preprocessing time) and the query time.
It is easy to see that any algorithms for NNS can also be used to solve the Closest Pair problem, 
as we can try each vector in $S$ as the query vector against the remaining vectors in $S$, 
and output the pair with minimum distance.  

Most early work on this problem is for fixed dimension. 
Indeed, when $m=1$ the problem is easy, as we can just sort the input vectors (which in this case are numbers), then perform a binary search to find the closest vector to the input query. 
For $m \geq 2$, Clarkson~\cite{Cla88} gave an algorithm with query time polynomial in $m \log n$, 
and space complexity $O(n^{\left \lceil{m/2}\right \rceil})$. 
Meiser~\cite{Mei93} designed an algorithm which 
runs in $O(m^5 \log n)$ time and uses $O(n^{m+\epsilon})$ space for arbitrary $\epsilon > 0$. 
By far, all efficient data structures for NNS have dimension $m$ appear in the exponent of the space complexity, 
due to the \emph{curse of dimensionality}.

This motivated people to introduce a relaxed version of Nearest Neighbor Search called 
the $(1+\epsilon)$-\emph{Approximate Nearest Neighbor Search}  ($(1+\epsilon)$-\emph{Approximate} NNS) Problem in the 1990s.
The problem now is, for an input query point $q$, find a point $p$ in $S$ such that  
the Hamming distance is: \\
\[
\dist{p,q} \leq (1+\epsilon)\min_{p'\in S}\dist{p',q}. 
\]
We call such a $p$ as a $(1+\epsilon)$-approximate nearest neighbor of input query $q$. \\ 
The $(1+\epsilon)$-Approximate NNS Problem has been studied extensively in the last two decades.
In 1998, Indyk and Motwani~\cite{IM98} used a set of hash functions to store the dataset such that if two points are close enough, 
they will have a very high probability to be hashed into the same buckets. 
As a pair of close points have higher probability than a pair of far-apart points to fall into the same bucket, 
the scheme is called \emph{locality sensitive hashing} (LSH). 
The query time of LSH is $O(n^{\frac{1}{1+\epsilon}})$, which is sublinear, 
and the space complexity of LSH is $O(n^{1+\frac{1}{1+\epsilon}})$, which is subquadratic. 
After Indyk and Motwani introducing the locality sensitive hashing, 
there have been many improvements on the parameters under different metric spaces, 
such as $\ell_p$ metric~\cite{KOR98, DIIM04, AI08, OWZ14, MNP06}. 
Recently, Andoni~\etal~\cite{ALRW17} gave tight upper and lower bounds on the time-space trade-offs of (data-dependent) hashing based algorithms for the $(1+\epsilon)$-Approximate NNS Problem. 
This is the first algorithm that achieves sublinear query time and near-linear space, 
for any $\epsilon > 0$. 
For many results on the Approximate NNS problem in high dimension, see e.g.~\cite{Ind04} for a survey. 
Some algorithms for the low dimension problem are surveyed in~\cite{AM05}.


In 2012 Valiant~\cite{Val15} leveraged fast matrix multiplication to obtain a new algorithm 
for the $(1+\epsilon)$-Approximate NNS Problem that is not based on LSH. \footnote{In fact, Valiant's algorithm can handle polynomially many ``outlier'' pairs.}
The general setting of Valiant's results is the following.
Suppose there is a set of points $S$ in $m$-dimensional Euclidean (or Hamming) space, 
and we are promised that for any $a \in S$ and $b \in S$,  $\langle a,b\rangle < \alpha$, 
except for only one pair which has $\langle a,b\rangle \geq \beta$ 
(which corresponds to the closest pair, and $\beta$ is known as the Pearson-correlation coefficient), 
for some $0<\alpha<\beta<1$.  
Valiant's algorithm finds the closest pair in 
$n^{\frac{5-\omega}{4-\omega}+\omega \frac{\log \beta}{\log \alpha}} m^{O(1)}$ time, 
where $\omega$ is the exponent for fast matrix multiplication ($\omega < 2.373$). 
Notice that, if the Pearson-correlation coefficient $\beta$ is some fixed constant, then 
when $\alpha$ approaches $0$ the running time tends to $n^{\frac{5-\omega}{4-\omega}}$, 
which is less than $n^{1.62}$. 
Valiant applied his algorithms to get improved bounds\footnote{All these
results are due to the fact that Valiant's algorithms are much more robust to weak correlations than other algorithms. Our algorithms therefore do not give improved bounds for these learning problems in the general settings.} 
for the Learning Sparse Parities with Noise Problem,
the Learning $k$-Juntas with Noise Problem, the Learning $k$-Juntas without Noise Problem, and so on.
More recently, Karppa~\etal~\cite{KKK16} improved upon Valiant's algorithm and obtained an algorithm that runs 
in $n^{\frac{2 \omega}{3}+O(\frac{\log \beta}{\log \alpha})}m^{O(1)}$ time. 

Note that, in general, algorithms for the $(1+\epsilon)$-Approximate NNS can only be applied to the gapped version of the Closest Pair Problem; 
for non-gapped version, as the minimum distant and the second minimum distant can differ by $1$, 
which means that the approximation parameter $\epsilon$ tends to zero if the minimum distance is large, 
the running time will approach to quadratic. However, our non-gapped version algorithm still runs in truly subquadratic time in this case.

\medskip
\paragraph{Decoding Random Binary Linear Codes.}
In 2015, May and Ozerov~\cite{MO15} observed that algorithms for high dimensional 
Nearest Neighbor Search Problem can be used to speedup the approximate matching part of the \emph{information
set decoding} algorithm. 
They designed a new algorithm for the \emph{Bichromatic Hamming Closest Pair} problem when the two input
lists of vectors are pairwise independent, 
and consequently obtained a decoding algorithm for random binary linear codes with time complexity $2^{0.097n}$. 
This improved upon the previously best result of Becker~\etal~\cite{BJM+12} which runs in $2^{0.102n}$.

\medskip
\paragraph{The Bichromatic Hamming Closest Pair problem.}
In fact, the problem studied in~\cite{AW15,ACW16,MO15} is the following
\emph{Bichromatic Hamming Closest Pair Problem}: 
we are given $n$ red vectors $R= \{r_0, r_1, \cdots , r_{n-1} \}$ and 
$n$ blue vectors $B=\{ b_0, b_1, \cdots , b_{n-1} \}$ from $\{0,1\}^m$, 
and the goal is to find a red-blue pair with minimum Hamming distance. 
It is easy to see that the Closest Pair Problem is reducible to the Bichromatic Hamming Closest Pair Problem
via a random reduction.
In fact, our algorithm for the Closest Pair Problem can also be easily adapted to solve the Bichromatic Hamming Closest Pair Problem 
as follows.
Run the decoding part of our algorithm on both sets $R$ and $B$ to get 
$\tilde{R} = \{\tilde{r}_0, \tilde{r}_1, \cdots , \tilde{r}_{n-1} \}$ 
and $\tilde{B} = \{\tilde{b}_0, \tilde{b}_1, \cdots , \tilde{b}_{n-1} \}$,  
sort $\tilde{R}$ and $\tilde{B}$ separately (without comparing the original vectors 
for adjacent pairs in the sorted lists), 
then merge the two sorted lists into one, and compute the distance between the original vectors 
for each red-blue pair of vectors that are compared during the merging process.
On the other hand, the Bichromatic Closest Pair Problem is unlikely to have \emph{truly} subquadratic algorithms
under some mild conditions. 
Assuming the Strong Exponential Time Hypothesis (SETH), for any $\eps>0$, there exists a constant $c$ such that when the dimension $m=c\log{n}$, then there is no $2^{o(m)} \cdot n^{2-\epsilon}$-time algorithm for the Bichromatic Closest Pair Problem~\cite{AW15, ARW17, Wil18}.

\subsection{Organization}



The rest of the paper is organized as follows.
Preliminaries and notations that we use throughout the paper 
are summarized in Section~\ref{sec:prelim}. 
In Section~\ref{sec:main} we present our main decoding-based algorithms for the Closest Pair Problem, 
assuming the minimum pairwise distance is given.
We then show how to get rid of this assumption in Section~\ref{sec:search}.
In Section~\ref{sec:light_bulb}, we apply our new algorithms to study the Light Bulb Problem.
Finally, we conclude with several open problems in Section~\ref{sec:conclusion}.

\section{Preliminaries}\label{sec:prelim}
Let $m\geq 1$ be a natural number, we use $[m]$ to denote the set $\{1,\ldots, m\}$.
All logarithms in this paper are base 2 unless specified otherwise. 

The binary entropy function, denoted $H_{2}(p)$, is defined as $H_{2}(p) := -p\log{p}-(1-p)\log(1-p)$
for $0\leq p \leq 1$. 

Let $\F_{q}$ be a finite field with $q$ elements\footnote{When $q=2$, we use
$\F_2$ and $\{0,1\}$ interchangeably throughout the paper.} and $m\geq 1$ be a natural number. 
If $x \in \F_{q}^{m}$ is an $m$-dimensional vector over $\F_q$ and $i\in [m]$, then we use $(x)_i$ to denote
the $\ith{i}$ coordinate of $x$.
The \emph{Hamming distance} between two vectors $x, y\in \F_{q}^{m}$ is the number of coordinates at which they differ:
$\dist{x,y}=|\{i \in [m]: (x)_i \neq (y)_i\}|$. 
For a vector $x \in \F^{m}$ and a real number $r\geq 0$, the Hamming ball of radius $r$ around $x$ is 
$B(x,r)=\{y\in \F^{m}: \dist{x,y}\leq r\}$.
The \emph{weight} of a vector $x$, denoted $\wt{x}$, is the number of coordinates at which $(x)_i\neq 0$.
The distance between two vectors $x$ and $y$ is easily seen to be equal to $\wt{x-y}$.

We also need the following bounds on binomial coefficients, see e.g. ~\cite[p. 309]{MS77}.
\begin{lemma}\label{lem:binomial}
Let $n$ be a natural number and $\lambda n$ be an integer, where $0<\lambda<1$. Then
\[
\frac{1}{\sqrt{8n\lambda(1-\lambda)}}2^{nH_2(\lambda)} \leq \binom{n}{\lambda n} \leq \frac{1}{\sqrt{2\pi n\lambda(1-\lambda)}}2^{nH_2(\lambda)}.
\]

\end{lemma} 

\subsection{Error correcting codes}

\begin{definition}[Error correcting codes]
Let $\F_q$ be a finite field with $q$ elements\footnote{In fact, error correcting codes, as well as constructing new codes out of 
existing codes by concatenations to be discussed shortly, can be defined more generally 
over an arbitrary set of $q$ distinct elements called \emph{alphabet} of the code. 
For the purpose of designing algorithms in this paper,
restricting to finite fields is simpler and sufficient.} and let $m\geq 1$ be a natural number.
A subset $C$ of $\F_q^{m}$ is called an $(m, K, d)_{q}$-\emph{code} if $|C|=K$ and for any two distinct vectors 
$x, y\in C$, $\dist{x, y}\geq d$. 
The vectors in $C$ are called \emph{codewords} of $C$, $m$ the block length of $C$, 
and $d$ the \emph{minimum distance} of $C$.
\end{definition}

Normalized by the block length $m$, $\kappa(C) := (\log_{q}{K})/m$ is known as the \emph{rate} of $C$
and $\delta(C) := d/m$ is known as the \emph{relative distance} of $C$.
If $C$ is a linear subspace of $\F_q^{m}$ of dimension $k$, the code is called a \emph{linear code} 
and denoted by $[m,k,d]_{q}$.
It is convenient to view such a linear code as the image of an \emph{encoding function} $E: \F_q^{k} \to \F_q^{m}$,
and $k$ is called \emph{message length} of $C$.
This can be generalized to non-linear codes as well where we view 
$\lfloor \log_{q}K \rfloor$ as the effective message length.
We usually drop the subscript $q$ when $q=2$.

\begin{definition}[Covering radius]
Let $C \subseteq \F_q^{m}$ be a code. For any $x\in \F^{m}$, 
define the distance between $x$ and $C$ to be
$\dist{x, C} := \min_{y\in C}\dist{x,y}$ (clearly, $\dist{x,C}=0$ if and only if $x$ is a codeword of $C$).
The \emph{covering radius} of a code $C$, denoted $R(C)$, 
is defined to be the maximum distance of any vector in $\F_q^{m}$ from $C$,
i.e., $R(C)=\max_{x\in \F_q^{m}}\dist{x,C}$.
\end{definition}

\subsubsection{Unique decoding}
Given an $(m, K, d)$-code $C$, if a vector (aka received word) $x\in \F_q^m$ is at a distance 
$r\leq \lfloor \frac{d-1}{2} \rfloor$ from some codeword $w$ in $C$, 
then by triangle inequality, $x$ is closer to $c$ than any other codewords in $C$. Therefore $x$
can be \emph{uniquely} decoded to the codeword $c\in C$. Such a decoding scheme\footnote{
Strictly speaking, the procedure described here is \emph{error correcting} instead of \emph{decoding},
where the latter should return the inverse of codeword $c$ of the encoding function.} 
is called \emph{unique decoding} (or minimum distance decoding) of code $C$,
and we shall call $\lfloor \frac{d-1}{2} \rfloor$ the (unique) \emph{decoding radius} of $C$.

\subsubsection{Gilbert-Varshamov bound and Gilbert's greedy code}
The Gilbert-Varshamov bound asserts that there is an infinite family of codes $C$ (essentially random codes 
or even random linear codes meet this bound almost surely) that satisfy 
$\kappa(C)\geq 1-H_2(\delta(C))$.
In particular, the following greedy algorithm of Gilbert~\cite{Gil52} finds a (non-linear) binary code $C$
of block length $m$ and minimum distance $d$ and 
satisfies that $\frac{1}{m}\log{K} \geq 1-H_{2}(d/m)-\eps$ for any $\eps>0$ for all sufficiently large $m$.
Start with $S=\F_2^m$ and $C=\emptyset$; while $S \neq \emptyset$, pick any element $x\in S$, add it to $C$
and remove all the elements in $B(x,d)$ from $S$.
We denote such a code by $\GV_{m,d}$.

We will need the following simple facts about $\GV_{m,d}$
\begin{lemma}\label{lem:GV-code}
The greedy algorithm of Gilbert can be implemented to run in $O(2^m)$ time, and produces
a decoding lookup table that supports constant time unique decoding. 
That is, for any $x\in \F_2^m$, if there is a codeword $w\in \GV_{m,d}$ with 
$\dist{x,w}\leq \lfloor \frac{d-1}{2} \rfloor$,
then the lookup entry of $x$ is $w$; otherwise the entry is a special symbol, say, $\perp$.
Moreover, the code $\GV_{m,d}$ constructed by Gilbert's greedy algorithm
satisfies that $R(\GV_{m,d})\leq d$.
\end{lemma}

\subsubsection{Reed-Solomon codes} 
\begin{definition}[Reed-Solomon codes]
Let $\F_q$ be finite field, $k$ and $m$ be integers satisfying $k\leq m \leq q$.
The encoding function for \emph{Reed-Solomon code} from $\F_k$ to $\F_m$ is the following:
First pick $m$ distinct elements $\alpha_1, \ldots, \alpha_m \in \F_q$;
on input $(a_0, a_1, \ldots, a_{k-1})\in \F_q^k$, define a degree-$k-1$ polynomial $P:\F_q \to \F_q$ as
$P(x)=\sum_{i=0}^{k-1}a_{i}x^i$; 
finally output the evaluations of $P(x)$ at $\alpha_1, \ldots, \alpha_m$,
i.e. the codeword is $(P(\alpha_1), \ldots, P(\alpha_m))$.
We will denote such a code by $\RS_{q,m,k}$.
\end{definition}

\begin{theorem}
The Reed-Solomon code defined above is an $[m, k, m-k+1]_q$ linear code.
\end{theorem}

\begin{theorem}[\cite{WB86}]
There exists an efficient unique decoding algorithm for Reed-Solomon codes which runs in time $\poly{m,\log{q}}$.
\end{theorem}
Reed-Solomon codes are optimal in the sense that they meet the Singleton bound, which states that
for any linear $[m,k,d]_q$-code, $k\leq m-d+1$.

\subsubsection{Concatenated codes}
The most commonly used way to transform a nice code which has constant rate and constant relative distance over a large alphabet to a similarly nice code over binary is \emph{concatenation}, which was first introduced by Forney~\cite{For66}.

\begin{definition}[Concatenated codes]
Let $C_1$ be an $(m_1, K_1, d_1)_{Q}$-code and
let $C_2$ be an $(m_2, K_2, d_2)_{q}$-code with $K_2\geq Q$.
Then the code obtained by concatenating $C_1$ with $C_2$, denoted by $C=C_1 \diamond C_2$,
is an $(m,K,d)_q$-code defined as follows.
Let $\phi$ by any mapping from $\F_Q$ onto $C_2$. 
Then the codewords of $C_1 \diamond C_2$ are obtained by replacing
each element in $\F_Q$ of any codeword $w=((w)_1, \ldots, (w)_{m_1}) \in C_1$ with the corresponding codeword in $C_2$
defined by $\phi$; namely 
$C =\{\phi((w)_1) \circ \cdots \circ \phi((w)_{m_1}) : w\in C_1\}$,
where each $\phi((w)_j)$ consists of $m_2$ elements in $\F_q$ and $\circ$ denotes string concatenation.
Note that each codeword in $C$ is an element in $\F_q^{m_1 m_2}$ and there are $K_1$ such codewords, 
therefore $m=m_1 m_2$ and $K=K_1$.
Usually $C_1$ is called the \emph{outer code} and $C_2$ is called the \emph{inner code}.
\end{definition}

It is well-known that the minimum distance of $C$ is $d_1 d_2$,
and the rate of $C$ is $\kappa(C)=\kappa(C_1)\kappa(C_2)$. 
Another useful fact is that $C$ can be efficiently decoded as long as both $C_1$ and $C_2$
can be efficiently decoded.


\begin{fact}\label{fact:CC_decoding}
Suppose $C_1$ is an $(m_1, K_1, d_1)_Q$-code with a decoding algorithm $A_1$ running in
$p_1(m_1, \log{Q})$ time,
$C_2$ is an $(m_2, K_2, d_2)_q$-code, where $K_2\geq Q$, and a decoding algorithm $A_2$ running in
$p_2(m_2, \log{q})$ time.
If $C$ is the concatenated code $C=C_1 \diamond C_2$, and then there is a decoding algorithm $A$ for $C$ which
run in time $p_1(m_1, \log{Q})+ m_1 p_2(m_2, \log{q})$ by first decoding $m_1$ received words of $C_2$ 
each consisting of $m_2$ elements in $\F_q$,
and then decode the $m_1$ concatenated elements in $\F_Q$ as a received word of $C_1$. 
\end{fact}


\subsubsection{Codes used in our algorithms}\label{sec:our_codes}
Some of the codes to be employed in our algorithm are a family of codes constructed by 
concatenating Reed-Solomon codes with certain binary non-linear Gilbert's greedy codes 
meeting the Gilbert-Varshamov bound.
It is well-known that concatenated codes such constructed can be made to meet the so-called 
\emph{Zyablov bound}\footnote{In fact, a stronger bound called \emph{Blokh-Zyablov bound} 
can be achieved by applying multilevel concatenations 
(see e.g.~\cite{Dum98} for a detailed discussion on multilevel concatenations of codes); however, as 
the improvement is minor, we only use single level concatenation in our code constructions 
to make the algorithms simpler.}
\begin{align}\label{eqn:kappa_Z}
\kappa(C) \geq \max_{0<\kappa(C_2) < 1-H_{2}(\delta(C))}
\kappa(C_2)\left(1-\frac{\delta(C)}{H_{2}^{-1}(1-\kappa(C_2))} \right)
\end{align}

Suppose we want a binary $(m, K, d)$-code for our algorithms, where $m$ and $d$ are fixed and
our goal is to maximize $K$, conditioned on that the code is efficiently decodable.
We pick a Reed-Solomon code $C_1=\RS_{q,m_1,k_1}$ and a Gilbert's greedy code $C_2 =\GV_{m_2,d_2}$
with the following constraints: $m_1 m_2 \leq m$ ($m_1 m_2$ should be as close to $m$ as possible), 
$d_1 d_2 \geq d$, $K_2=2^{m_2\kappa(C_2)} \geq q > m_1$,  and
$2^{m_2} \leq \poly{m_1}$. It is easy to check that there are large ranges of values for $m_1$ and $m_2$,
and optimizing the choice of $d_2$ (and therefore $\delta(C_2)$) makes our concatenated code 
$C=C_1 \diamond C_2$ both meets the Zyablov bound in Eqn.~\eqref{eqn:kappa_Z}
and can be decoded in $\poly{m}$ time.

We will denote the maximum rate as a function of the relative distance $\delta$ given by the Zyablov bound   
by $\kappa_{\mathsmaller{Z}}(\delta)$, and similarly denote the maximum rate given by the Gilbert-Varshamov bound
by $\kappa_{\mathsmaller{GV}}(\delta)$ (i.e. $\kappa_{\mathsmaller{GV}}(\delta)=1-H_{2}(\delta)$).
Note that $\kappa_{\mathsmaller{Z}}(\delta) \leq \kappa_{\mathsmaller{GV}}(\delta)$ for all $0 \leq \delta \leq 1/2$,
and the reason we use codes achieving only $\kappa_{\mathsmaller{Z}}(\delta)$ is because such codes can be 
generated and decoded in $\poly{m}$ time. 

\subsection{The Closest Pair Problem}
Given $n$ vectors $x_0, x_1, \ldots, x_{n-1}$ in $\{0,1\}^{m}$,
the Closest Pair Problem is to find two vectors whose pairwise Hamming distance is minimum.
For ease of exposition and without loss of generality, we will assume throughout the paper that there is a unique pair, namely $x_0$ and $x_1$, that achieves the minimum pairwise distance $\dmin$.
We will use $d_2$ to denote the second minimum pairwise distance, where $d_2\geq \dmin+1$.
In the most general case, we do not make any assumption about $m$, $\dmin$ or $d_2$.

\section{Main Algorithm for the Closest Pair Problem}\label{sec:main}
We now present our Main Algorithm for the Closest Pair Problem. For ease of exposition,
we make a somewhat unnatural assumption that the value of $\dmin$ is given. 
However, as we show in Section~\ref{sec:search}, the algorithm can be modified
to get rid of this assumption, with only a slight slowdown in running time. 

\begin{theorem}[Non-gapped version]\label{thm:general_rand}
Let $x_0, x_1, \ldots, x_{n-1}$ in $\{0,1\}^{m}$ be $n$ binary vectors such that 
$x_0$ and $x_1$ is the unique pair achieving the minimum pairwise distance $\dmin$ (and the second smallest
distance can be as small as $\dmin + 1$). 
Suppose we are given the value of $\dmin$ and let $\delta := \dmin/m$.
Then there is a randomized algorithm
running in $O(n\log^{2}n\cdot 2^{(1-\kappa_{\mathsmaller{Z}}(\delta)-\delta)m} \cdot \poly{m})$ 
which finds the closest pair $x_0$ and $x_1$ with probability at least $1-1/n^2$.  
\end{theorem}
\begin{proof}
Our Main Algorithm for the Closest Pair problem is described in Algorithm~\ref{alg:main}, 
and the decoding subroutine
$\textrm{Dec}(C, r, x)$ is illustrated in Algorithm~\ref{alg:dec}.
Note that we choose the minimum distance of $C$ to be $\dmin + 1$, hence the decoding radius of $C$
is $\dmin/2$ (without loss of generality, assume that $\dmin$ is even).

For the correctness of the algorithm, first note that our algorithm will 
output the correct minimum distance if and only if
$x_0$ is ever compared against $x_1$ for computing pairwise distance, and this happens if and only if $x_0$ and $x_1$
are adjacent in the sorted array after decoding. A sufficient condition for the latter is that
the decoded vectors of $x_0$ and $x_1$ are identical and they are different from any other decoded vectors. 

How many shift vectors $y\in \{0,1\}^m$ in Algorithm~\ref{alg:main} satisfy this condition?
We will call such vectors \emph{good} vectors.
Denote the set of vectors lying at the ``middle'' between $x_0$ and $x_1$ by 
\[
\textsc{mid}=\{z\in \{0,1\}^m: \dist{x_0, z} = \dist{z, x_1} = \dmin/2\}.
\]
Note that any vector $y$ that shifts a vector $z\in \textsc{mid}$ to a codeword $c\in C$ would be a \emph{good}
vector. To see this, first note that after such a shift, $y+z$ is a codeword in $C$,
and both $y+x_0$ and $y+x_1$ lie within the decoding radius of $y+z$, and therefore will be decoded 
to $y+z$. Moreover, the shifted vector of any other input vector $y+x_i$, $2 \leq i \leq n-1$, 
lies outside the decoding radius of $y+z$.
This is because if it does, then by triangle inequality and 
the fact that the decoding radius of $C$ is $\dmin/2$, 
\begin{align*}
\dist{x_0, x_i} 
&=\dist{y+x_0, y+x_i} \\
&\leq \dist{y+x_0, y+z} +\dist{y+z, y+x_i} \\
&\leq \dmin/2+\dmin/2 = \dmin,
\end{align*}
contradicting our assumption that $x_0$ and $x_1$ is the unique pair achieving the minimum distance.

How many such \emph{good} vectors?
There are in total $\binom{\dmin}{\dmin/2}$ vectors exist in $\textsc{mid}$, and all their pairwise distances
are at most $\dmin$. Let $c_1, c_2$ be two distinct codewords in $C$. By our choice of the minimum distance of $C$,
$\dist{c_1, c_2} > \dmin$. Consider any two distinct vectors $z_1$ and $z_2$ in $\textsc{mid}$.
Clearly applying these two shift vectors to the same codeword gives two distinct vectors,
namely $c_1+z_1$ and $c_1+z_2$. 
Moreover, applying two distinct vectors in $\textsc{mid}$ to two distinct codewords also results in two 
distinct shift vectors, because
\[
\dist{c_1+z_1, c_2+z_2}=\wt{c_1+c_2+z_1+z_2} > 0,
\]
since $\wt{c_1+c_2}\geq d >\dmin$ but $\wt{z_1+z_2} = \dist{z_1,z_2}\leq \dmin$.

Recall that $C$ is a $(m, K, d)$-code and hence there are $K$ codewords in $C$. 
It follows that there are in total $K\cdot \binom{\dmin}{\dmin/2}$ \emph{good} vectors of this kind.
Therefore 
\[
\Pr(\text{a random $y$ succeeds in finding the closest pair}) 
\geq \frac{K\cdot \binom{\dmin}{\dmin/2}}{2^m}, 
\]
and hence repeatedly selecting 
\begin{align*}
 2\ln{n}\frac{2^m}{K\cdot \binom{\dmin}{\dmin/2}} 
&= O\left(\log{n}\frac{\sqrt{\delta m}2^m}{2^{\kappa_{\mathsmaller{Z}}(\delta) m}2^{\delta m}}\right) \\
&= O(2^{(1-\kappa_{\mathsmaller{Z}}(\delta)-\delta)m}m^{1/2}\log{n})
\end{align*} 
independent $y$'s
will succeed with probability at least $1-1/n^2$, where in the last step we use the bound $\binom{n}{n/2}=O(\frac{2^n}{\sqrt{n}})$, a special case of Lemma~\ref{lem:binomial}.

Finally, note that each choice of shift vector $y$ requires $n\cdot \poly{m}$ time decoding
as well as $O(n\log{n}\cdot m)$ sorting and comparing adjacent vectors, so the total running time of the algorithm
is $O(n\log^{2}n\cdot 2^{(1-\kappa_{\mathsmaller{Z}}(\delta)-\delta)m} \cdot \poly{m})$.
\end{proof}

\RestyleAlgo{boxruled}
\LinesNumbered
\begin{algorithm}[ht] 
\caption{Main Algorithm for the Closest Pair Problem}\label{alg:main}
\SetKwData{Left}{left}\SetKwData{This}{this}\SetKwData{Up}{up}
\SetKwFunction{Decode}{Decode}
\SetKwInOut{Input}{input}\SetKwInOut{Output}{output}
\Input{A set of $n$ vectors $x_0,  \ldots, x_{n-1}$ in $\{0,1\}^{m}$ and $\dmin$ }
\Output{Two vectors and their distance}
\BlankLine
	generate a binary $(m, K, d)$-code $C$ with $d = \dmin +1$ \\
	\For{$i\leftarrow 1$ \KwTo $ O(2^{(1-\kappa_{\mathsmaller{Z}}(\delta)-\delta)m}m^{1/2}\log{n})$}{
		pick a random $y\in \{0,1\}^{m}$ \\
		\For{$j\leftarrow 0$ \KwTo $n-1$}{\label{forins}
			$\tilde{x}_{j} \leftarrow \mathrm{Dec}(C, \lfloor \dmin/2 \rfloor, y+x_j) $
		}
		sort $\tilde{x}_{0}, \ldots, \tilde{x}_{n-1}$ \\ 
		(suppose the sorted sequence is $\tilde{x}_{s_0}, \ldots, \tilde{x}_{s_{n-1}}$, where $\{s_0, \ldots, s_{n-1} \}$ is
			a permutation of $\{0,1, \ldots n-1\}$) \\
		\For{$j\leftarrow 1$ \KwTo $n-1$}{\label{forins2}
			compute $\dist{x_{s_{j-1}},x_{s_j}}$
			}
	}
	output the pair of vectors with minimum distance ever found and their distance \
\end{algorithm}

\RestyleAlgo{boxruled}
\LinesNumbered
\begin{algorithm}[ht] 
\caption{$\mathrm{Dec}(C, r, x)$} \label{alg:dec}
	\SetKwData{Left}{left}\SetKwData{This}{this}\SetKwData{Up}{up}
	\SetKwFunction{Decode}{Decode}
	\SetKwInOut{Input}{input}\SetKwInOut{Output}{output}
	\Input{A binary $(m, K, d)$-code $C$, a decoding radius $r < d/2$, and a vector $x \in \{0,1\}^{m}$}
	\Output{A vector $\tilde{x} \in \{0,1\}^{m}$}
	\BlankLine
	run the (efficient) decoding algorithm for $C$ on input vector $x$, and let the output vector be $\tilde{x}$ \\
	\eIf{$\dist{x, \tilde{x}} \leq r$} 
		{output $\tilde{x}$}
	{output $x$}

\end{algorithm}

If we assume further that a decoding algorithm for some binary $(m, K, d)$-code $C$ which meets
the Gilbert-Varshamov bound is given as a black box, then the running time in
Theorem~\ref{thm:general_rand} can be improved to 
$O(n\log^{2}n\cdot 2^{(H_{2}(\delta)-\delta)m} \cdot \poly{m})$.
Note that this is not a totally unrealistic assumption, as for most interesting settings,
$m=c\log{n}$ for some small constant $c$.\footnote{As in the settings of random vectors, e.g. the Light Bulb Problem,
$m=c\log{n}$ is both necessary and sufficient to distinguish $n$ stochastic bit sequences.} Therefore, greedily searching for a binary code of 
block length $m$ that meets the Gilbert-Varshamov bound is tantamount to 
running an $O(n^c)$ time preprocessing, which can be reused for any problem instance with the same
vector length and minimum closest pair distance. 

If there is a gap between $d_2$ and $\dmin$ (this roughly corresponds to the \emph{approximate
closest pair} problem in~\cite{Val15}), then we can improve the running time of the Main Algorithm
in Theorem~\ref{thm:general_rand} by exploiting an error correcting code with larger decoding radius. 
\begin{theorem}[Gapped version]\label{thm:general_rand2}
Let $x_0, x_1, \ldots, x_{n-1}$ in $\{0,1\}^{m}$ be $n$ binary vectors such that 
$x_0$ and $x_1$ is the unique pair achieving the minimum pairwise distance $\dmin$.
Suppose we are given the values of $\dmin$ as well as the second minimum distance $d_2$.
Let $\delta := \dmin/m$ and $\delta' := d_2/m$.
Then there is a randomized algorithm
running in $O(n\log^{2}n\cdot 2^{(1-\kappa_{\mathsmaller{Z}}(\delta')
-\delta-(1-\delta)H_{2}(\frac{\delta' - \delta}{2(1-\delta)}))m} \cdot \poly{m})$ 
which finds the closest pair $x_0$ and $x_1$ with probability at least
$1-1/n^2$.  
Moreover, the running time can be further improved to 
$O(n\log^{2}n\cdot 2^{(H_2(\delta')
-\delta-(1-\delta)H_{2}(\frac{\delta'-\delta}{2(1-\delta)}))m} \cdot \poly{m})$, 
if we are given black box access to the decoding algorithm of an $(m, K, d)$-code which meets
the Gilbert-Varshamov bound.
\end{theorem}

\begin{proof}
The proof follows a similar structure as the proof of Theorem~\ref{thm:general_rand}.
The main difference is now we pick a binary error correcting code of minimum distance $d_2+1$,
thereby decoding radius $r=d_2/2=\frac{1}{2}\delta' m$ (once again, for simplicity, we assume $d_2$ is even).

Accordingly, the ``middle point'' set is now defined as  
\[
\textsc{mid}_G =\{z\in \{0,1\}^m: \text{$\dist{x_0, z} \leq r$ and $\dist{x_1, z} \leq r$}\}.
\]
We now give a lower bound on the size of $\textsc{mid}_G$.

Without loss of generality, we assume $x_0=0^m$ and let $T=\{i \in [m]: (x_1)_i = 1\}$.
Clearly $|T|=\dmin$. Let $i = |\{k\in T: (z)_k = 0\}|$
and $j = |\{k\in [m]\setminus T: (z)_k = 1\}|$. Then
$\dist{x_0, z} \leq r$ is equivalent to $\dmin-i+j\leq r$,
and $\dist{x_1, z} \leq r$ is equivalent to $i+j\leq r$.
Therefore 
\begin{align*}
|\textsc{mid}_G|
&=\sum_{i+j\leq r}\sum_{\dmin-i+j\leq r} \binom{\dmin}{i} \binom{m-\dmin}{j} \\
&\geq \binom{\dmin}{\dmin/2} \binom{m-\dmin}{r-\dmin/2} \\
&=\Theta\left( \frac{2^{\delta m}}{\sqrt{\delta m}} 
\frac{2^{(1-\delta)H_{2}(\frac{\delta' - \delta}{2(1-\delta)})m}}{\sqrt{(1-\delta)m}}\right),
\end{align*}
where the last step follows from Lemma~\ref{lem:binomial}. The rest of the proof is identical to that of Theorem~\ref{thm:general_rand}, and therefore is omitted.
\end{proof}

\subsection{A deterministic variant of the Main Algorithm}
One can turn our randomized Main Algorithm into a deterministic one by exhaustively searching for
all possible shift vectors $y \in \F_2^m$. A simple observation is that it suffices to check 
for all vectors in the Hamming ball of radius equals to the covering radius of the code $C$.

\begin{theorem}\label{thm:general_det}
Let $x_0, x_1, \ldots, x_{n-1}$ in $\{0,1\}^{m}$ be $n$ binary vectors such that 
$x_0$ and $x_1$ is the unique pair achieving the minimum pairwise distance $\dmin$ (and the second smallest
distance can be as small as $\dmin + 1$). 
Suppose we are given the value of $\dmin$ and let $\delta := \dmin/m$.
Then there is a deterministic algorithm that finds the closest pair $x_0$ and $x_1$ with
running time $O(n\log{n}\cdot 2^{H_{2}(1-\kappa_{\mathsmaller{Z}}(\delta))m} \cdot \poly{m})$.
Moreover, if we are given as black box the decoding algorithm of a random Varshamov linear code
with block length $m$ and minimum distance $\dmin+1$, then the running time is
$O(n\log{n}\cdot 2^{H_{2}(H_{2}(\delta))m} \cdot \poly{m})$.
\end{theorem}

\begin{proof}
Let $\delta := \dmin/m$.
It is well-known that for any linear $[m,k,d]_q$-code $C$, the covering radius of $C$ satisfies that
$R(C)\leq m-k$. It follows that for Reed-Solomon code $\RS_{q,m,k}$, $R(\RS)\leq m-k < d$.
We can either generate a random \emph{linear} Varshamov code~\cite{Var57}
similar to that described in Section~\ref{sec:our_codes} that meets the Gilbert-Varshamov bound
and concatenate it with a Reed-Solomon code so that the resulting binary code is a linear code.
Then the covering radius of this concatenated code satisfies that 
$R(C)\leq (1-\kappa_{\mathsmaller{Z}}(\delta))m$.
Or, if preprocessing is allowed, we may simply generate a random linear Varshamov code of block length $m$,
whose covering radius satisfies that $R(C)\leq (1-\kappa_{\mathsmaller{GV}}(\delta))m=H_2(\delta)m$. 

Now the deterministic algorithm for finding the closest pair is similar to the Main Algorithm,
except that instead of picking random shift vector $y$, the algorithm checks every 
$y\in B(0^m, R(C))$. It follows directly that the running time of the algorithm is
$O(n\log{n}\cdot \poly{m} \cdot \mathrm{Vol}(B(0^m, R(C))))$.
Here $\mathrm{Vol}(B(0^m, R(C)))$ denotes the number of vectors within the Hamming ball $B(0^m, R(C))$,
which is $2^{H_{2}(1-\kappa_{\mathsmaller{Z}}(\delta))m}$ for the concatenated code, 
or $2^{H_{2}(H_2(\delta))m}$ for the random Varshamov linear code.

The correctness of the algorithm follows that, by the same argument of the correctness of 
Algorithm~\ref{alg:main}, any vector $z \in \textsc{mid}$ is at most $R(C)$ away from some codeword $c\in C$,
namely $\dist{z, c}=\wt{z+c} \leq R(C)$.
When vector $z+c$, which lies in  $B(0^m, R(C))$, is chosen as the shift vector $y$,
$x_0$ and $x_1$ will be the only two vectors decoded to $c$, therefore the algorithm successfully 
finds the closest pair. 
\end{proof}

We remark that our covering radius argument seems to be too rough, 
as there are many vectors in $\textsc{mid}$. Getting a more efficient deterministic algorithm,
or derandomizing the Main Algorithm is an interesting open question of combinatorics in nature.

\section{Searching for the Minimum Distance}\label{sec:search}
In this section we show how to remove the assumption that 
the value of $\dmin$ is given to the Main algorithm,
Basically we show that one can use a binary-search like procedure to find $\dmin$ 
without too much slowdown of the Main Algorithm.
Our key observation is that, although the decoding radius is chosen to be $\dmin/2$ in the Main Algorithm,
actually we can relax this requirement: indeed, any decoding radius between $\dmin/2$ and $\dmin$ works.

\begin{lemma}\label{lem:search}
The Main Algorithm works (with worse running time) as long as the binary error correcting code used  
has decoding radius $r =\lfloor \frac{d-1}{2} \rfloor$ satisfying $\frac{1}{2}\dmin \leq r \leq \dmin$.
\end{lemma}
\begin{proof}
The proof is similar to the proof of Theorem~\ref{thm:general_rand},
but we slightly generalize the original definition of \textsc{mid} as follows.
Let
\[
\textsc{mid}_1 = \{z\in \{0,1\}^m: \text{$\dist{x_0, z} = r$ and $\dist{x_1, z} = \dmin-r$}\},
\]
let
\[
\textsc{mid}_2 = \{z\in \{0,1\}^m: \text{$\dist{x_1, z} = r$ and $\dist{x_0, z} = \dmin-r$}\}.
\]
and finally let $\textsc{mid}' = \textsc{mid}_1 \cup \textsc{mid}_2$.

Clearly the set $\textsc{mid}'$ is non-empty. 
The key point is that any vector $y$ that shifts some vector $z\in \textsc{mid}'$ to a codeword $c\in C$ 
must be a \emph{good} vector, following a similar argument as in the proof of Theorem~\ref{thm:general_rand}.
The running time of the algorithm can also be calculated similarly.
\end{proof}

\RestyleAlgo{boxruled}
\LinesNumbered
\begin{algorithm}[ht]
\caption{Searching for $\dmin$}\label{alg:search}
	\SetKwData{Left}{left}\SetKwData{This}{this}\SetKwData{Up}{up}
	\SetKwFunction{Decode}{Decode}
	\SetKwInOut{Input}{input}\SetKwInOut{Output}{output}
	\Input{A set of $n$ vectors $x_0,  \ldots, x_{n-1}$ in $\{0,1\}^{m}$}
	\Output{The minimum pairwise distance $\dmin$ }
	\BlankLine
	$r\leftarrow 1$ \\
	\While{true}{
		run Algorithm 2 with $d=r$\;
		\eIf{the minimum distance $\dmin$ returned in Algorithm~\ref{alg:main} is at most $2r$ \\}{
			\Return $\dmin$\;
			}{
			$r\leftarrow (1+\eps)r$\;
			}
		}
\end{algorithm}

Our algorithm for finding $\dmin$ is illustrated in Algorithm~\ref{alg:search}.
The correctness of Algorithm~\ref{alg:search} follows from two simple facts:
first, Algorithm~\ref{alg:main} can never return a value less than $\dmin$;
second, when $\dmin/2 \leq r \leq \dmin$, by Lemma~\ref{lem:search} and Theorem~\ref{thm:general_rand},
Algorithm~\ref{alg:main} returns the correct value of $\dmin$ (with high probability). 

In fact, to make our algorithm more efficient, for any $0<\eps \leq 1$, we can 
search with decoding radius $r = 1, \lfloor (1+\eps) \rfloor, \lfloor (1+\eps)^2 \rfloor, \cdots$.
Note that by Lemma~\ref{lem:search}, the maximum value we will ever try is $(1+\eps)\dmin/2$.
As the running time of Algorithm~\ref{alg:main} is monotone increasing with respect to 
the decoding radius $r$, so in order to bound the running time of searching for $\dmin$,
it suffices to bound the running time of Algorithm~\ref{alg:main} for $r=(1+\eps)\dmin/2$. 
Following a similar analysis as in the proof of Theorem~\ref{thm:general_rand},
\[
|\textsc{mid}'| \geq |\textsc{mid}_1|=\binom{\dmin}{(1+\eps)\dmin/2}.
\]
Therefore,
\begin{align*}
&\Pr(\text{a random $y$ succeeds in finding the closest pair}) \\
\geq & \frac{K\cdot \binom{\dmin}{(1+\eps)\dmin/2}}{2^m} \\
= &\Theta\left(\frac{2^{\kappa_{\mathsmaller{Z}}((1+\eps)\delta) m}2^{H_{2}(\frac{1-\eps}{2})\delta m}}
{\sqrt{\delta m}2^m} \right),
\end{align*}
where in the last step we use bounds in Lemma~\ref{lem:binomial}. As the binary search calls at most $\log_{(1+\eps)}\dmin < \log_{(1+\eps)} m = O(\eps^{-1}\log{m})$ 
times Algorithm~\ref{alg:main}, we therefore have the following theorem.
\begin{theorem}
Let $x_0, x_1, \ldots, x_{n-1}$ in $\{0,1\}^{m}$ be $n$ binary vectors such that 
$x_0$ and $x_1$ is the unique pair achieving the minimum pairwise distance $\dmin$.
Then for any $\eps>0$, there is a randomized algorithm
running in :\\
\[O(\eps^{-1}n\log^{2}n\cdot 
2^{(1-\kappa_{\mathsmaller{Z}}((1+\eps)\delta)-\delta H_{2}(\frac{1-\eps}{2}))m} \cdot \poly{m}).\]
which finds the $\dmin$ (as well as the closest pair $x_0$ and $x_1)$ with probability at least
$1-1/n^2$.
The running time can be improved to :\\
\[O(\eps^{-1}n\log^{2}n\cdot 
2^{(H_{2}((1+\eps)\delta)-\delta H_{2}(\frac{1-\eps}{2}))m} \cdot \poly{m}),\] 
if we are given black-box decoding algorithms for an ensemble of $O(\eps^{-1}\log m)$
binary error correcting codes that meet the Gilbert-Varshamov bound.
\end{theorem}

\section{The Light Bulb Problem}\label{sec:light_bulb}
In this section, we apply our new algorithms for the Closest Pair Problem to a special case of it,
namely the \emph{Light Bulb Problem}. 

In the Light Bulb Problem, we are given $n$ sequences of bit strings $X_0, X_1, \ldots,$ $ X_{n-1}$.
All bits are generated independently, uniformly at random from $\{0,1\}$, except that
two strings, say $X_0$ and $X_1$, are generated with non-zero linear correlation $\rho$;
that is, independently for each $i$, $\Pr((X_0)_i = (X_1)_i)=\frac{1+\rho}{2}$
and $\Pr((X_0)_i \neq (X_1)_i)=\frac{1-\rho}{2}$. 
The problem is to find this correlated pair of sequences.

First note that we may assume the Pearson correlation $\rho$ is positive, 
as there is a simple randomized reduction 
from the negative $\rho$ case to the positive $\rho$ case:
given an instance of the Light Bulb Problem with $\rho<0$
randomly pick $n/2$ sequences and flip all the bits in these sequences.
Then with probability $1/2$, the correlated pair become $-\rho$ correlated.

To apply our algorithms for the Closest Pair Problem to the Light Bulb Problem,
the following standard result\footnote{This is a folklore bound. Indeed, similar analyses can be
found in earlier work, although sometimes with slightly different tools (e.g.,
Hoeffding bound in place of Chernoff bound), but essentially they all
aim to show that with a sufficiently high dimension, the planted
correlation is unique, with high probability.} provides a randomized reduction from the latter to the former.
We include a proof for completeness and to justify our choice of the dimension $m$.

\begin{theorem}\label{thm:sample_size}
If we pick $m=\frac{4\ln2\cdot \log n}{\rho^2}(1+o(1))$ bits at random from $X_0, X_1, \ldots,$ $ X_{n-1}$
to obtain $n$ vectors $x_0, x_1, \ldots, x_{n-1}$ in $\{0,1\}^m$,
then with constant probability, $x_0$ and $x_1$ is the \emph{unique}
closest pair among these $n$ vectors.
\end{theorem}

\begin{proof}
For each pair of vectors $x_i$ and $x_j$, $0\leq i < j \leq n-1$,
define $m$ indicator random variables $\{(I_{i,j})_k\}_{k\in[m]}$
such that $(I_{i,j})_k = 1$ if and only if $(x_i)_k \neq (x_j)_k$.
Note that for any pair $i<j$, $\{(I_{i,j})_k\}_{k\in[m]}$ are $m$
independent and identically distributed random variables, and
$\dist{x_i, x_j}=\sum_{k\in [m]}(I_{i,j})_k$.
Specifically, $\Pr((I_{0,1})_k = 0) = \frac{1+\rho}{2}$
and $\Pr((I_{0,1})_k = 1) = \frac{1-\rho}{2}$;
and $\Pr((I_{i,j})_k = 0) = \Pr((I_{i,j})_k = 1) = 1/2$ for all other $i<j$ pairs.

Note that each pairwise distance $\dist{x_i,x_j}$ is a \emph{binomial} random variable.
In particular, $\dist{x_0, x_1}$ is a $B(m, \frac{1-\rho}{2})$ random variable
and all others are $B(m,1/2)$ random variable. To argue about the distribution of distance between $x_0$ and $x_1$, we need the following fact:
\begin{fact}[\cite{KB80}]\label{fact:binom} 
Binomial distribution $B(n,p)$ has median $\left \lfloor np \right \rfloor$ or $\left \lceil np \right \rceil$.
\end{fact}
\noindent Let $d_t := \Exp(\dist{x_0, x_1})= (1-\rho)m/2$.
Then by Fact~\ref{fact:binom}, $\Pr(\dist{x_0, x_1} \geq d_t)\leq 1/2$.

On the other hand, for any other pair $x_i$ and $x_j$,
\begin{align*}
\Pr(\dist{x_i, x_j} < d_t)
&= \Pr\left(\dist{x_i, x_j}< \Exp(\dist{x_i, x_j})-\rho m/2 \right) \\
&< e^{-(\rho m)^{2}/2m}=e^{-m\rho ^{2}(1-o(1))/2} \\
&\leq \frac{1}{2n^2},
\end{align*}
by a simple application of the Chernoff bound (e.g. Theorem A.1.1 in~\cite{AS08}).
Now applying a union bound over all $x_i$ and $x_j$ pairs,
we have that with probability at least $1/4$,
$\dist{x_0, x_1} < d_t$ and for all other pairs $\dist{x_i, x_j}\geq d_t$,
i.e., $x_0$ and $x_1$ is the unique closest pair among these $n$ vectors.
\end{proof}

Note that Theorem~\ref{thm:sample_size} implies that if we sample 
$m=\frac{4\ln2\cdot \log n}{\rho^2}(1+o(1))$ bits from the $n$ random sequences,
then with constant probability, we get an instance of the Closest Pair Problem with 
$\dmin < (1-\rho)m/2$.
Now, by repeatedly running our randomized algorithm for Closest Pair Problem $O(\log{n})$ times,
each time taking independent samples from the input vectors, and then take a majority vote, then
by combining a simple application of the Chernoff bound, Theorem~\ref{thm:general_rand} and Theorem~\ref{thm:sample_size}, we obtain the following
\begin{theorem}
There is a randomized algorithm for the Light Bulb Problem which runs in time
\[
O(n\cdot \poly{\log n})\cdot 2^{(1-\kappa_{\mathsmaller{Z}}(\frac{1-\rho}{2})-\frac{1-\rho}{2})
\frac{4\ln2\cdot \log n}{\rho^2}(1+o(1))}
\]
and succeeds with probability at least
$1-1/n^2$. 
The running time can be further improved to 
\[
O(n\cdot \poly{\log n})\cdot 2^{(H_{2}(\frac{1-\rho}{2})-\frac{1-\rho}{2})
\frac{4\ln2\cdot \log n}{\rho^2}(1+o(1))},
\]
if we are allowed a one-time preprocessing time of $n^{2.773/\rho^2}$ to generate the decoding
lookup table of a random Gilbert's $(m,K,(1-\rho)m/2)$-code.
\end{theorem}
Numerical calculations show that our new algorithm performs
better than the improved Valiant's fast matrix multiplication algorithm~\cite{KKK16}
(which runs in $n^{1.582}$) when 
$\rho \geq 0.9967$ (equivalently when $\delta\leq 0.00165$). 
Moreover, if an $n^{2.773/\rho^2}$-time preprocessing is allowed, then 
our algorithm runs faster for all $\rho \geq 0.9310$ (equivalently for all $\delta \leq 0.0345$).

\medskip
\paragraph{Deterministic algorithm.}
Following~\cite{KKKO16}, we say a deterministic algorithm solves the Light Bulb Problem if it is correct \emph{on almost all instances}, i.e., if the algorithm fails on a randomly picked instance with probability at most $1/\poly{n}$.
Following a similar proof that of the randomized algorithm shown before, we have the following theorem on deterministic algorithm for the Light Bulb Problem 
\begin{theorem}
There is a deterministic algorithm for the Light Bulb Problem which runs in time
\[
O(n\cdot \poly{\log n})\cdot 2^{H_{2}(1-\kappa_{\mathsmaller{Z}}(\frac{1-\rho}{2}))\frac{4\ln2\cdot \log n}{\rho^2}(1+o(1))}
\]
and succeeds with probability at least
$1-1/n^2$. 
The running time can be further improved to 
\[
O(n\cdot \poly{\log n})\cdot 2^{H_{2}(H_{2}(\frac{1-\rho}{2}))\frac{4\ln2\cdot \log n}{\rho^2}(1+o(1))},
\]
if we are allowed a one-time preprocessing time of $n^{2.773/\rho^2}$ to generate the decoding
lookup table of a random Gilbert's $(m,K,(1-\rho)m/2)$-code.
\end{theorem}
Note that, like the randomized algorithm, our deterministic algorithm also needs to draw $O(m\log{n})$ bits from each of the $n$ sequences. However, the algorithm uses no random bits and the success probability is over the random instance the algorithm gets from the input.

As mentioned earlier, Alman~\cite{Alm19} gave the currently best deterministic algorithm for the Light Bulb Problem,
which runs in $O(n^{1.582})$ time. 
Unsurprisingly, the deterministic version of our algorithm outperforms the one in~\cite{Alm19} when the Pearson
correlation is very large. 
Specifically, by numerical calculation, our deterministic algorithm runs faster than Alman's when $\rho \geq 0.999948$.
Moreover, if an $n^{2.773/\rho^2}$-time preprocessing is allowed, then 
when $\rho \geq 0.9933$ (equivalently when $\delta \leq 0.0033$),
we may take the vector length $m\leq 2.8101 \log n$ so that 
our deterministic algorithm runs in at most $O(n^{1.581})$.


\section{Concluding Remarks and Open Problems}\label{sec:conclusion}
We propose a simple approach, namely a decoding-base method, to solve the classic Closest Pair Problem.
Our results leave open several interesting questions. The way we derandomize our randomized algorithm is by a simple brute-force search. Is there a smarter and more efficient way to derandomize? 
Valiant's fast matrix multiplication method~\cite{Val15} for the Light Bulb Problem 
is the only known algorithm that makes good use of the availability of larger amount of data. 
Is it possible to leverage the data size to improve
the running time of our decoding approach?
Another interesting open question is to study the Closest Pair Problem in the streaming model,
as many real-life situations of the problem --- such as in cyber security --- 
are in fact in this setting.

\section*{Acknowledgements}
We are most grateful to the anonymous referees for their detailed and invaluable comments and suggestions. We would like to thank Karthik C.S. for his comments.

\bibliographystyle{plain}

\begin{thebibliography}{99}

\bibitem{ARW17}
A.~Abboud, A.~Rubinstein, and R.~Williams.
\newblock Distributed {PCP} theorems for hardness of approximation in {P}.
\newblock In {\em Proc.\ 58th Annual IEEE Symposium on Foundations of Computer
  Science}, pages 25--36, 2017.

\bibitem{AWY15}
A.~Abboud, R.~Williams, and H.~Yu.
\newblock More applications of the polynomial method to algorithm design.
\newblock In {\em Proc.\ 26th {ACM}-{SIAM} {S}ymposium on {D}iscrete
  {A}lgorithms}, pages 218--230, 2015.

\bibitem{Alm19}
J.~Alman.
\newblock An illuminating algorithm for the light bulb problem.
\newblock In {\em Proc.\ 2nd {S}ymposium on {S}implicity in {A}lgorithms},
  pages 2:1--2:11, 2019.

\bibitem{ACW16}
J.~Alman, Timothy~M. Chan, and R.~Williams.
\newblock Polynomial representations of threshold functions and algorithmic
  applications.
\newblock In {\em Proc.\ 57th Annual IEEE Symposium on Foundations of Computer
  Science}, pages 467--476, 2016.

\bibitem{AW15}
J.~Alman and R.~Williams.
\newblock Probabilistic polynomials and {H}amming nearest neighbors.
\newblock In {\em Proc.\ 56th Annual IEEE Symposium on Foundations of Computer
  Science}, pages 136--150, 2015.

\bibitem{AS08}
N.~Alon and J.~Spencer.
\newblock {\em The Probabilistic Method}.
\newblock John Wiley and Sons, third edition, 2008.

\bibitem{AI08}
A.~Andoni and P.~Indyk.
\newblock Near-optimal hashing algorithms for approximate nearest neighbor in
  high dimensions.
\newblock {\em Communications of the ACM}, 51:117--122, 2008.

\bibitem{Ind04}
A.~Andoni and P.~Indyk.
\newblock Nearest neighbors in high-dimensional spaces.
\newblock In J.~Goodman, J.~O'Rourke, and Csaba~D. Toth, editors, {\em Handbook
  of Discrete and Computational Geometry}. Chapman and Hall/CRC, 3rd edition,
  2017.

\bibitem{ALRW17}
A.~Andoni, T.~Laarhoven, I.~Razenshteyn, and E.~Waingarten.
\newblock Optimal hashing-based time-space trade-offs for approximate near
  neighbors.
\newblock In {\em Proc.\ 28th {ACM}-{SIAM} {S}ymposium on {D}iscrete
  {A}lgorithms}, pages 47--66, 2017.

\bibitem{AM05}
S.~Arya and D.~Mount.
\newblock Computational geometry: Proximity and location.
\newblock In D.~P. Mehta and S.~Sahni, editors, {\em Handbook of Data
  Structures and Applications}. Chapman and Hall/CRC, 2005.

\bibitem{ARL+05}
C.~Aston, D.~Ralph, D.~Lalo, S.~Manjeshwar, B.~Gramling, D.~DeFreese, A.~West,
  D.~Branam, L.~Thompson, M.~Craft, et~al.
\newblock Oligogenic combinations associated with breast cancer risk in women
  under 53 years of age.
\newblock {\em Human genetics}, 116(3):208--221, 2005.

\bibitem{BJM+12}
A.~Becker, A.~Joux, A.~May, and A.~Meurer.
\newblock Decoding random binary linear codes in $2^{n/20}$: How $1+1=0$
  improves information set decoding.
\newblock In {\em Proc.\ 31st Annual {I}nternational {C}onference on the
  {T}heory and {A}pplications of {C}ryptographic {T}echniques}, pages 520--536,
  2012.

\bibitem{BB15}
A.~Benson and G.~Ballard.
\newblock A framework for practical parallel fast matrix multiplication.
\newblock In {\em Proceedings of the 20th ACM SIGPLAN Symposium on Principles
  and Practice of Parallel Programming}, pages 42--53, 2015.

\bibitem{Ben80}
J.~Bentley.
\newblock Multidimensional divide-and-conquer.
\newblock {\em Communications of the ACM}, 23(4):214--229, 1980.

\bibitem{CNG+98}
J.~Cho, D.~Nicolae, L.~Gold, C.~Fields, M.~LaBuda, P.~Rohal, M.~Pickles,
  L.~Qin, Y.~Fu, J.~Mann, et~al.
\newblock Identification of novel susceptibility loci for inflammatory bowel
  disease on chromosomes 1p, 3q, and 4q: evidence for epistasis between 1p and
  {IBD}1.
\newblock {\em Proceedings of the National Academy of Sciences},
  95(13):7502--7507, 1998.

\bibitem{Cla88}
K.~Clarkson.
\newblock A randomized algorithm for closest-point queries.
\newblock {\em SIAM Journal on Computing}, 17(4):830--847, 1988.

\bibitem{Cor09}
H.~Cordell.
\newblock Detecting gene $\times$ gene interactions that underlie human
  diseases.
\newblock {\em Nature Reviews Genetics}, 10(6):392--404, 2009.

\bibitem{DIIM04}
M.~Datar, N.~Immorlica, P.~Indyk, and V.~Mirrokni.
\newblock Locality-sensitive hashing scheme based on $p$-stable distributions.
\newblock In {\em Proc.\ 20th {S}ymposium on {C}omputational {G}eometry}, pages
  253--262, 2004.

\bibitem{Dub10}
M.~Dubiner.
\newblock Bucketing coding and information theory for the statistical
  high-dimensional nearest-neighbor problem.
\newblock {\em IEEE Transactions on Information Theory}, 56(8):4166--4179,
  2008.

\bibitem{Dum98}
I.~Dumer.
\newblock Concatenated codes and their multilevel generalizations.
\newblock In V.~Pless and W.~C. Huffman, editors, {\em Handbook of coding
  theory}, volume~2, pages 1911--1988. Amsterdam, The Netherlands: Elsevier
  Science, 1998.

\bibitem{For66}
D.~Forney.
\newblock {\em Concatenated Codes}.
\newblock MIT Press, 1966.

\bibitem{FBC+07}
K.~Frazer, D.~Ballinger, D.~Cox, D.~Hinds, L.~Stuve, R.~Gibbs, J.~Belmont,
  A.~Boudreau, P.~Hardenbol, S.~Leal, et~al.
\newblock A second generation human haplotype map of over 3.1 million {SNP}s.
\newblock {\em Nature}, 449(7164):851--861, 2007.

\bibitem{Gil52}
E.~Gilbert.
\newblock A comparison of signalling alphabets.
\newblock {\em Bell System Technical Journal}, 31(3):504--522, 1952.

\bibitem{GRSS95}
M.~J. Golin, R.~Raman, C.~Schwarz, and M.~Smid.
\newblock Simple randomized algorithms for closest pair problems.
\newblock {\em Nordic Journal of Computing}, 2(1):3--27, 1995.

\bibitem{GRS18}
V.~Guruswami, A.~Rudra, and M.~Sudan.
\newblock Essential coding theory.
\newblock Book draft in preparation, available at
  \url{https://cse.buffalo.edu/faculty/atri/courses/coding-theory/book/web-coding-book.pdf},
  2018.

\bibitem{HRM+17}
J.~Huang, L.~Rice, D.~Matthews, and R.~van~de Geijn.
\newblock Generating families of practical fast matrix multiplication
  algorithms.
\newblock In {\em Parallel and Distributed Processing Symposium (IPDPS), 2017
  IEEE International}, pages 656--667, 2017.

\bibitem{IM98}
P.~Indyk and R.~Motwani.
\newblock Approximate nearest neighbors: Towards removing the curse of
  dimensionality.
\newblock In {\em Proc.\ 30th Annual ACM Symposium on the Theory of Computing},
  pages 604--613, 1998.

\bibitem{KB80}
R.~Kaas and J.~Buhrman.
\newblock Mean, median and mode in binomial distributions.
\newblock {\em Statistica Neerlandica}, 34(1):13--18, 1980.

\bibitem{KKK16}
M.~Karppa, P.~Kaski, and J.~Kohonen.
\newblock A faster subquadratic algorithm for finding outlier correlations.
\newblock In {\em Proc.\ 27th {ACM}-{SIAM} {S}ymposium on {D}iscrete
  {A}lgorithms}, pages 1288--1305, 2016.

\bibitem{KKKO16}
M.~Karppa, P.~Kaski, J.~Kohonen, and P.~\'{O} Cath\'{a}in.
\newblock Explicit correlation amplifiers for finding outlier correlations in
  deterministic subquadratic time.
\newblock In {\em Proc.\ 24th Annual {E}uropean {S}ymposia on {A}lgorithms},
  pages 52:1--52:17, 2016.

\bibitem{KM95}
S.~Khuller and Y.~Matias.
\newblock A simple randomized sieve algorithm for the closest-pair problem.
\newblock {\em Information and Computation}, 118(1):34--37, 1995.

\bibitem{KOR98}
E.~Kushilevitz, R.~Ostrovsky, and Y.~Rabani.
\newblock Efficient search for approximate nearest neighbor in high dimensional
  spaces.
\newblock In {\em Proc.\ 30th Annual ACM Symposium on the Theory of Computing},
  pages 614--623, 1998.

\bibitem{Leg12}
F.~LeGall.
\newblock Faster algorithms for rectangular matrix multiplication.
\newblock In {\em Proc.\ 53rd Annual IEEE Symposium on Foundations of Computer
  Science}, pages 514--523, 2012.

\bibitem{MS77}
F.J. MacWilliams and N.~J.~A. Sloane.
\newblock {\em The Theory of Error-correction Codes}.
\newblock North Holland, 1977.

\bibitem{MO15}
A.~May and I.~Ozerov.
\newblock On computing nearest neighbors with applications to decoding of
  binary linear codes.
\newblock In {\em Proc.\ 34th Annual {I}nternational {C}onference on the
  {T}heory and {A}pplications of {C}ryptographic {T}echniques}, pages 203--228,
  2015.

\bibitem{Mei93}
S.~Meiser.
\newblock Point location in arrangements of hyperplanes.
\newblock {\em Information and Computation}, 106(2):286--303, 1993.

\bibitem{MNP06}
R.~Motwani, A.~Naor, and R.~Panigrahi.
\newblock Lower bounds on locality sensitive hashing.
\newblock In {\em Proc.\ 22nd {S}ymposium on {C}omputational {G}eometry}, pages
  154--157, 2006.

\bibitem{MSL+07}
S.~Musani, D.~Shriner, N.~Liu, R.~Feng, C.~Coffey, N.~Yi, H.~Tiwari, and
  D.~Allison.
\newblock Detection of gene $\times$ gene interactions in genome-wide
  association studies of human population data.
\newblock {\em Human heredity}, 63(2):67--84, 2007.

\bibitem{OWZ14}
R.~O'Donnell, Y.~Wu, and Y.~Zhou.
\newblock Optimal lower bounds for locality-sensitive hashing (except when $q$
  is tiny).
\newblock {\em ACM Transactions on Computation Theory}, 6(1), 2014.

\bibitem{Pan18}
V.~Pan.
\newblock Fast feasible and unfeasible matrix multiplication, April 2018.
\newblock \url{https://arxiv.org/abs/1804.04102}.

\bibitem{PRR95}
R.~Paturi, S.~Rajasekaran, and J.~Reif.
\newblock The light bulb problem.
\newblock {\em Information and Coputation}, 117:187--192, 1995.

\bibitem{Rab76}
M.~Rabin.
\newblock Probabilistic algorithms.
\newblock {\em Algorithms and Complexity}, pages 21--30, 1976.

\bibitem{Smi97}
M.~Smid.
\newblock Closest-point problems in computational geometry.
\newblock In J.~Sack and J.~Urrutia, editors, {\em Handbook of computational
  geometry}, pages 877--935. Elsevier Science Publishing, 1997.

\bibitem{Sto10}
A.~Stothers.
\newblock {\em On the complexity of matrix multiplication}.
\newblock PhD thesis, The University of Edinburgh, 2010.

\bibitem{Val15}
G.~Valiant.
\newblock Finding correlations in subquadratic time, with applications to
  learning parities and juntas.
\newblock {\em Journal of the ACM}, 62(13), 2015.
\newblock Earlier version in FOCS'12.

\bibitem{Val88}
L.~G. Valiant.
\newblock Functionality in {N}eural {N}ets.
\newblock In {\em First Workshop on Computational Learning Theory}, pages
  28--39, 1988.

\bibitem{Var57}
R.~R. Varshamov.
\newblock Estimate of the number of signals in error correcting codes.
\newblock {\em Dokl. Akad. Nauk SSSR}, 117(5):739--741, 1957.

\bibitem{WB86}
L.~Welch and E.~Berlekamp.
\newblock Error correction for algebraic block codes, December~30 1986.
\newblock US Patent 4,633,470.

\bibitem{Wil14a}
R.~Williams.
\newblock Faster all-pairs shortest paths via circuit complexity.
\newblock In {\em Proc.\ 46th Annual ACM Symposium on the Theory of Computing},
  pages 664--673, 2014.

\bibitem{Wil14b}
R.~Williams.
\newblock The polynomial method in circuit complexity applied to algorithm
  design (invited talk).
\newblock In {\em Conference on Foundation of Software Technology and
  Theoretical Computer Science (FSTTCS)}, pages 47--60, 2014.

\bibitem{Wil18}
R.~Williams.
\newblock On the difference between closest, furthest, and orthogonal pairs:
  Nearly-linear vs barely-subquadratic complexity.
\newblock In {\em Proc.\ 29th {ACM}-{SIAM} {S}ymposium on {D}iscrete
  {A}lgorithms}, pages 1207--1215, 2018.

\bibitem{Wil12}
V.~Williams.
\newblock Multiplying matrices faster than {C}oppersmith-{W}inograd.
\newblock In {\em Proc.\ 44th Annual ACM Symposium on the Theory of Computing},
  pages 887--898, 2012.

\end{thebibliography}

\end{document}